\documentclass[12pt]{article}
\usepackage{fullpage,graphicx,psfrag,verbatim,titlesec}
\usepackage[toc,page]{appendix}
\usepackage[small,bf]{caption}
\usepackage{amsfonts}
\usepackage{printlen}
\usepackage{arydshln}
\usepackage{xcolor}
\usepackage{centernot}
\usepackage{mathtools}
\usepackage{multirow}
\usepackage{amsthm}
\usepackage{bm}
\usepackage{cases}
\usepackage{enumitem}
\usepackage[utf8]{inputenc}
\usepackage[english]{babel}

\usepackage{apacite}
\usepackage{authblk}
\usepackage{xr}

\usepackage[nottoc]{tocbibind}

\newtheorem{example}{Example}
\newtheorem{remark}{Remark}
 \newtheorem{corollary}{Corollary}
 \newtheorem{lemma}{Lemma}

\usepackage{amssymb}
\usepackage[normalem]{ulem}

\title{Identifiability of Latent Class Models with Covariates}
\author{Jing Ouyang and Gongjun Xu\\
Department of Statistics, University of Michigan}
\date{}
\begin{document}
\maketitle

\begin{abstract}
	Latent class models with covariates are widely used for psychological, social, and educational research. Yet the fundamental identifiability issue of these models has not been fully addressed. Among the previous research on the identifiability of latent class models with covariates, \citeA[{\it Psychometrika}, \textbf{69}:5-32]{huang2004} studied the local identifiability conditions. However, motivated by recent advances in the identifiability of the restricted latent class models, particularly  Cognitive Diagnosis Models (CDMs), we show in this work that the conditions in \citeA{huang2004}  are only necessary but not sufficient to determine the local identifiability of the model parameters. 
	 To address the open identifiability issue for latent class models with covariates, this work establishes conditions to ensure the global identifiability of the model parameters in both strict and generic senses. Moreover, our results extend to the polytomous-response CDMs with covariates, which generalizes the existing identifiability results for CDMs.
	 
\end{abstract}

\newpage

\section{Introduction}
Latent class models are extensively applied in numerous scientific fields,  including educational assessments, biological research, and psychological measurements,  to 
 infer the latent subgroups of a population as well as each subject's latent classification information. For instance, one application of latent class models in cognitive diagnosis is to classify individuals with different latent attributes based on their observed responses to items, for which reason they are key components in educational measurements \cite{junker2001, von2008toefl}, psychiatric evaluations \cite{Templin2006}, and disease detections \cite{wu2017}. 
In addition to understanding the basic parameters in latent class models, researchers are also interested in studying the relations between latent class parameters with the observed covariates, 
such as subjects'  gender, race, education level, and other characteristics \cite{Formann1985,collins2009latent,huang2004}. 

Latent class models with covariates can help to improve the classification accuracy of the latent classes and are useful in testing  whether the covariates are related to the latent class membership probability or response probability. Such latent class models involving covariates have been studied in many works in psychometrics and statistics literature, where covariates were mostly constrained  to be discrete at early stage \cite{clogg1984,Formann1985}, and further extended to be in general forms \cite{Dayton1988, vander1996, muthen2017mplus}. The models have been popularly applied in educational, psychological, and behavioral sciences \cite{collins2009latent,muthen2005discrete,reboussin2008locally,bakk2013,park2018explanatory}.
The related estimation problems  have also received great interest from researchers in the psychometrics field, such as estimating the covariate coefficients  \cite{Janne2012}, adjusting for the bias in the estimation \cite{bakk2013}, and  estimating the number of latent classes  \cite{huang2005,pan2014}.

For latent class models with or without   covariates, identifiability is one of the most fundamental issues as it is the prerequisite for parameter estimations and statistical inferences. Identifiability could be interpreted as the feasibility of recovering the model parameters based on observed responses, i.e., the  parameters in identifiable models should be distinct given the probabilistic distribution of the observations. A rich body of literature have studied identifiability issues, dating back to \citeA{koopmans1950} and \citeA{koopmansetal1950}. Specifically, \citeA{McHugh1956} proposed conditions to determine the local identifiability for the binary-response latent class models, and  \citeA{goodman1974} further extended the local identifiability conditions to the polytomous-response models. In the sense of strict identifiability, \citeA{gyllenberg1994} found that the binary-response latent class models can not be strictly identifiable. Nonetheless, \citeA{allman2009} considered the concept of generic identifiability and established sufficient conditions for the generic identifiability of latent class models, where a model is said to be generically identifiable if the model parameters are identifiable except for a measure-zero set of parameters. However, their generic identifiability conditions can only be applied to the unrestricted latent class models, but not directly to the restricted latent class models. To address this issue, \citeA{xu2017} and \citeA{xu2018JASA} established the results for the identifiability of the $Q$-restricted binary-response latent class models. For the polytomous-response models, \citeA{Culpepper2019} and \citeA{Fang2019} established strict identifiability conditions based on the algebraic theorems proposed by \citeA{kruskal1977}. Moreover, \citeA{gu2018partial} studied the generic and partial identifiability of the restricted binary-response latent class models and extended their conditions to the polytomous-response models as well.

Among existing research, most focus on the identifiability of general or restricted latent class models without covariates, whereas few investigate the identifiability of latent class models with covariates. As the observed covariates represent  characteristics of certain homogeneous groups, incorporating covariates into latent class models would help to explain the association of these characteristics with latent classes. The regression latent class models with covariates are general extensions of  latent class models without covariates. In other words, the regular or restricted latent class models can be viewed as a special family of latent class models with covariates, where all covariates values are zero. 
Technically speaking, existing identifiability results for regular or restricted latent class models cannot be directly applied to the regression latent class models due to the existence of covariates, and new techniques are needed to establish the identifiability of the corresponding regression coefficients for those covariates, which do not exist in the regular or restricted latent class models. In the literature, \citeA{huang2004} was among the first to study the identifiability of latent class models with covariates. The authors studied the local identifiability conditions for the  model parameters, that is, the conditions to ensure that the model parameters are identifiable in a neighborhood of the true parameters.

\begin{sloppypar}
	However, as to be shown in the paper, the proposed identifiability conditions in \citeA{huang2004} are only necessary but not sufficient for the local identifiability of latent class models. Our argument borrows ideas from the recent developments in the identifiability of Cognitive Diagnosis Models (CDMs), a special family of the restricted latent class models. 
Besides, the results in \citeA{huang2004} only concern the local identifiability but not the global identifiability. 
In light of these, our work establishes identifiability conditions to check the global identifiability for  latent class models with covariates. 
Furthermore, we also establish the identifiability results for CDMs with covariates, which is a special family of the regression latent class models. Our results extend  many identifiability conditions for the binary-response CDMs to the polytomous-response CDMs with covariates, and these conditions are beyond results in the existing literature related to CDMs identifiability \cite{xu2017,Culpepper2019, gu2018partial}.
\end{sloppypar}

The organization of this paper is as follows. Section \ref{Models and Existing Works} introduces the setup of the regression latent class models with covariates  as well as the regression CDMs, and reviews some existing identifiability results. Section \ref{Necessity and Sufficiency} discusses the necessity and sufficiency of the existing identifiability conditions for the regression latent class models. Section \ref{proposed identifiability conditions} presents our main results  for both strict and generic identifiability of the regression latent class models as well as the regression CDMs. Section \ref{Data Example} uses a Trends in Mathematics and Science Study (TIMSS) dataset  as an example to illustrate the application of the identifiability results in educational assessments. Section \ref{Conclusion} gives a discussion. The proofs for the main theorems and propositions are provided in the Supplementary Material.

\section{Model Setup and Existing Works}
\label{Models and Existing Works}
\subsection{Regression Latent Class Models (RegLCMs)}
\label{Covariate Effect Model}

We  start with the setup of latent class models without covariates.
 Suppose there are $N$ subjects responding to $J$ items. The response of subject $i$ is denoted as $\bm{R}_i = (R_{ij};j = 1, \ldots,J )$, where $R_{ij}$ denotes the response of subject $i$ to item $j$, for $ i = 1,\dots, N$. And $R_{ij}\in\{0, \dots, M_j-1\}$, where $M_j$ denotes the number of possible values for $R_{ij}$. Denote $\mathcal{S} = \bigtimes_{j=1}^{J} \{0,\dots,M_j-1\}$ as the set of all response patterns, and its cardinality is denoted as $S = |\mathcal{S}| = \prod_{j=1}^{J} M_j$. The case at $M_j=2$ corresponds to the binary-response models. Consider there are $C$ latent classes and denote $L_i$ as the latent class membership for subject $i$. Assume the $N$ subjects are independent; for  $c = 0, \dots, C-1$,  $L_i = c$ implies that the subject $i$ is in the $c$th latent class category and $ \eta_c = P(L_i = c)$ defines the latent class membership probability, i.e. the probability for subject $i$ being in the $c$th latent class. The latent class membership probabilities are summarized as $\bm{\eta} = (\eta_c; c= 0, \ldots,C-1)$. For any $ j =1, \dots, J$, $r = 0,\dots,M_j-1,$ and $c = 0,\dots, C-1$, we use $\theta_{jrc} = P( R_{ij} = r \mid L_i = c )$ to denote the conditional response probability, i.e. the probability of the response to item $j$ being $r$ given the subject $i$ is in the $c$th latent class. Let the vector $\bm{\theta_{jc}}=(\theta_{j0c}, \cdots, \theta_{j(M_j-1)c})$ to denote the probability vector for item $j$ given the latent class membership $c$. The conditional response probabilities are summarized as $\bm{\Theta} = (\bm{\theta_{jc}}; j=1,\ldots, J, c=0,\ldots, C-1)$. 
 The conditional probability mass function for $R_{ij}$ is   
$	P(R_{ij} \mid L_i=c, \bm{\theta_{jc}}) = \prod_{r = 0}^{M_j-1 }\theta_{jrc}^{\mathbb{I}\{R_{ij} = r\}}$, and the probability mass function
   of $\bm{R}_i$  is
\begin{align*}
	P(\bm{R}_i \mid  \bm{\eta},   \bm{\Theta}) &= \sum_{c=0}^{C-1}\eta_c  \prod_{j = 1}^{J} P(R_{ij} \mid   L_i = c , \bm{\theta_{jc}}) =  \sum_{c=0}^{C-1}\eta_c \prod_{j = 1}^{J} \prod_{r = 0}^{M_j -1 }\theta_{jrc}^{ \mathbb{I}\{R_{ij} = r\}}. 
\end{align*}

To introduce the regression latent class models, following the model setting  in \citeA{huang2004}, we let the latent class membership probability $\eta_c$'s and the conditional response probability $\theta_{jrc}$'s to be functionally dependent on covariates. Denote $(\bm{x}_i, \bm{z}_i)$ to be the covariates of subject $i$, where $\bm{x}_i = (1, x_{i1}, \cdots, x_{ip})_{(p+1)\times 1}^{T}$ are the primary covariates related to the latent class membership $\eta_c$ for $c = 0, \dots, C-1$, and $\bm{z}_i = (\bm{z}_{i1}, \cdots, \bm{z}_{iJ})_{J\times q}^{T}$ with $\bm{z}_{ij} = ( z_{ij1}, \cdots, z_{ijq})_{q\times1}^{T}$ are the secondary covariates associated with the conditional response probability $\theta_{jrc}$ for any $ j =1, \dots, J$, $r = 0,\dots,M_j-1,$ and $c = 0,\dots, C-1$.  The $x_{it}$ and $z_{ijs}$ can be categorical covariates representing gender, race or marital status. They can also be continuous, such as the subject's age. 
As in some applications, we may have certain prior knowledge on the set of the covariates related to $\eta_c$ and that of the  covariates related to $\theta_{jrc}$, where the two sets may or may not contain the same covariates.  
Hence we follow the general framework in \citeA{huang2004} by  applying different notations, $\bm{x}_i$ and $\bm{z}_i$,  to distinguish the covariates related to $\eta_c$ and $\theta_{jrc}$, while allowing the $\bm{x}_i$ and $\bm{z}_i$ to have some overlapped covariates.

Before presenting the generalized linear model framework, we need to clarify some notations. In models without covariates, e.g., latent class models or CDMs to be discussed in Section \ref{RegCDMs}, we use $\eta_c$ and $\theta_{jrc}$ to denote the corresponding latent class membership probability and conditional response probability, respectively. When covariates are involved in models, the parameters are  dependent on the covariates. In this situation, we denote $\eta_c^i = P(L_i = c \mid \bm{x}_i, \bm{z}_i)$ to be the latent class membership probability for subject $i$, and $\theta_{jrc}^i = P( R_{ij} = r \mid L_i = c ,\bm{x}_i, \bm{z}_i)$ to be the conditional response probability for subject $i$, for $i = 1, \ldots, N$.

Under the framework of generalized linear model, we use logit link function to relate $\eta_c^i$'s and $\theta_{jrc}^i$'s to covariates $(\bm{x}_i, \bm{z}_i)$. We let the log-odds be linearly dependent on the covariates and characterize the RegLCMs by the following equations
\begin{equation}
     \text{log}\Big( \frac{\eta_c^i}{\eta_0^i}\Big) = \beta_{0c} + \beta_{1c}x_{i1} + \cdots + \beta_{pc}x_{ip} \label{log eta}, 
\end{equation}
$\text{for} \ i = 1, \ldots, N, c = 1,\ldots, C-1$, and
\begin{equation}
	\text{log}\Big(\frac{\theta_{jrc}^i}{\theta_{j0c}^i}\Big) = \gamma_{jrc} + \lambda_{1jr}z_{ij1} + \cdots + \lambda_{qjr}z_{ijq}, \label{log p} 
\end{equation}
 $\text{for}\ i = 1, \ldots, N$, $j =1, \dots, J$, $r = 1, \cdots, {M_j-1}$ and $c = 0,\dots, C-1,$
 where $\beta,\gamma,\lambda$ are regression coefficient parameters. 
We want to point out that the identifiability conditions to be shown in Section 4 still hold for RegLCMs when the logarithmic function in \eqref{log eta} and \eqref{log p} is replaced with other monotonic functions. The key component in establishing the identifiability conditions for the coefficient parameters is the function monotonicity, which build the bijective mapping between identifiable $(\bm{\eta}, \bm{\Theta})$ and identifiable $(\bm{\beta}, \bm{\gamma}, \bm{\lambda})$. In this work, without loss of generality, we shall focus on the popularly used logit link function.

 From \eqref{log eta} and \eqref{log p}, we equivalently express $\eta_c^i$ and $\theta_{jrc}^i$ as
\begin{equation}
\label{eta} 
     {\eta_c^i} = \frac{\text{exp}(\beta_{0c} + \beta_{1c}x_{i1} + \cdots + \beta_{pc}x_{ip})}{1 + \sum_{l=1}^{C-1}\text{exp}(\beta_{0l} + \beta_{1l}x_{i1} + \cdots + \beta_{pl}x_{ip}) },
\end{equation}
for $ i = 1, \ldots, N$, $c = 0,\ldots, C-1$, and
\begin{equation}
\label{theta}
		\theta_{jrc}^i = \frac {\text{exp}{(\gamma_{jrc} + \lambda_{1jr}z_{ij1} + \cdots + \lambda_{qjr}z_{ijq}})}{ 1 + \sum_{s = 1}^{M_j - 1} \text{exp}(\gamma_{jsc} + \lambda_{1js}z_{ij1} + \cdots + \lambda_{qjs} z_{ijq})}, 
\end{equation}
$\text{for}\ i = 1, \ldots, N$, $j =1, \dots, J$, $r = 0, \cdots, {M_j-1}$ and $c = 0,\dots, C-1$.
From the above expressions, we  see that $\eta_c^i$ and $\theta_{jrc}^i$ are  functionally dependent on the linear functions $\bm{x}_i^{T} \bm{\beta}$ and $\bm{\gamma}_{jc} + \bm{z}_{ij}^{T} \bm{\lambda}_j$, where  $\bm{\beta} = (\bm{\beta}_{c}; c=0,\ldots, C-1)_{(p+1)\times C}$ with $\bm{\beta}_c = (\beta_{lc}; l=0,\ldots,p)^{T}_{(p+1)\times 1}$, $\bm{\gamma}_{jc} = (\gamma_{jrc}; r=0,\ldots, {M_j}-1)_{1 \times  M_j},$ and $\bm{\lambda}_j = (\bm{\lambda}_{jr}; r=0,\ldots, {M_j}-1)_{q\times M_j}$ with $\bm{\lambda}_{jr} = (\lambda_{ljr}; l=1,\ldots,q)^{T}_{q\times 1}$. 

Here following \citeA{huang2004}, in the conditional probability model \eqref{log eta}, the regression parameters ($\beta$) are latent class specific. 
In the conditional probability model \eqref{log p}, we allow the intercept parameters ($\gamma$) dependent on the latent class,  the  response level, and the item index, while the  regression coefficients parameters ($\lambda$) are dependent on the response level and the item index but not the latent class membership, which, as pointed in \citeA{huang2004}, is a logical assumption to prevent possible misclassification by adjusting for the associated covariates. The following two assumptions proposed by \citeA{huang2004} hold for all RegLCMs. 
\begin{enumerate}
	\item \label{assumption1} The latent class membership probability ${\eta}_c^i$ is dependent on $\bm{x}_i$ only and the conditional response probability ${\theta}_{jrc}^i$ is dependent on $\bm{z}_i$ only:
\begin{eqnarray}
P(L_i = c \mid \bm{x}_i, \bm{z_i}) &=& P(L_i = c \mid \bm{x}_i); \nonumber \\
P(R_{i1} = r_1, \cdots, R_{iJ} = r_J \mid L_i, \bm{x_i}, \bm{z_i} ) &=& P(R_{i1} = r_1, \cdots, R_{iJ} = r_J \mid L_i, \bm{z_i} ). \nonumber
\end{eqnarray}
	\item \label{assumption2} The measurements for different items are independent given the latent class and $\bm{z}_i$ (that is, the local independence assumption):
\begin{equation}
	P(R_{i1} = r_1, \cdots, R_{iJ} = r_J \mid L_i, \bm{z_i} ) = \prod_{j=1}^{J} P(R_{ij} = r_j \mid L_i, \bm{z_i} ). \nonumber
\end{equation}
\end{enumerate}

When the coefficients $\beta_{1c}, \cdots, \beta_{pc}$ in (\ref{log eta}) and $\lambda_{1jr}, \cdots, \lambda_{qjr}$ in (\ref{log p}) are zeros, RegLCMs will be reduced to latent class models without covariates, which is a special case in the family of RegLCMs. 
Next, we will introduce a special family of RegLCMs, Cognitive Diagnosis Models (CDMs), which is a family of the  restricted latent class models and has been substantially studied in educational and psychological measurement. From there we further introduce the regression CDMs. The two special RegLCMs (CDMs and regression CDMs) are important in the subsequent discussions about the identifiability conditions for RegLCMs.

\subsection{Cognitive Diagnosis Models as Special RegLCMs}
\label{RegCDMs}

In CDMs, each latent class corresponds to a distinct vector $\bm{\alpha} = (\alpha_{1}, \cdots, \alpha_{K}) \in \mathcal{A} = \{0,1\}^K$  where $\alpha_{1}, \cdots, \alpha_{K}$ denote $K$ binary latent attributes respectively and $\mathcal{A}$ denotes the attribute pattern space. The vector $\bm{\alpha}$ represents a unique latent profile with the $k$th entry $\alpha_{k} = 1$ implying the mastery of the subject on the $k$th latent attribute and $\alpha_{k} = 0$ implying his deficiency of it. The number of latent classes is $C = |\mathcal{A}| = 2^K$. For notational convenience, we follow the idea in \citeA{Culpepper2019} by introducing a tool vector $\bm{v} = (2^{K-1}, 2^{K-2}, \cdots, 1)^{T}$ and denote the latent class membership as $L = \bm{\alpha}^{T} \bm{v} =c \in \{0, \cdots, 2^K -1\}$. The key characteristics of CDMs is its introduction of the latent attributes and let the combinations of mastery or deficiency of each attribute to represent the latent class memberships in the restricted latent class models.

The relationship between the response $\bm{R} = (R_{1},\cdots,R_{J})$ and the attribute profile $\bm{\alpha}$ for any subject could be summarized through a binary matrix $Q_{J \times K}$.  Denote the $j$th row in $Q$-matrix to be $\bm{q}_j = (q_{j1},\cdots, q_{jK})$, where $q_{jk} \in \{0,1\}$ and  $q_{jk} = 1$ means that the $k$th attribute is required for subjects to solve item $j$. Similar to RegLCMs, we consider the general polytomous responses   $R_{j} \in \{0, \cdots, M_j-1\}$. Given a subject's latent profile $\bm{\alpha}$ with $\bm{\alpha}^{T}\bm{v} =c$, each $R_{j}$ follows a categorical distribution with the probability vector to be  $\bm{\theta_{jc}}=(\theta_{j0c}, \cdots, \theta_{j(M_j-1)c})$, where $\theta_{jrc} = P(R_{j} = r \mid  \bm{\alpha}^{T} \bm{v} =c )$ is the probability for getting response value $r$ in item $j$. The conditional probability mass function for $R_j$ is   
$	P(R_{j} \mid \bm{\alpha}^{T} \bm{v}=c,\ \bm{\theta_{jc}}) = \prod_{r = 0}^{M_j-1 }\theta_{jrc}^{\mathbb{I}\{R_{j} = r\}}$, and the probability mass function   for $\bm{R}$  is
\begin{align*}
	P(\bm{R} \mid \bm{\eta}, \bm{\Theta}) &= \sum_{c=0}^{2^K-1}   P( \bm{\alpha}^{T} \bm{v} =c)  \prod_{j = 1}^{J} P(R_{j} \mid    \bm{\alpha}^{T} \bm{v} = c , \bm{\theta_{jc}})  =\sum_{c=0}^{2^K-1}  \eta_c \prod_{j = 1}^{J} \prod_{r = 0}^{M_j -1 }\theta_{jrc}^{ \mathbb{I}\{R_{j} = r\}}. 
\end{align*}

Following the generalized DINA (G-DINA) model framework, we  decompose the log-odds of $\theta_{jrc}$ into a sum of attribute effects as follows.
This framework was introduced in \citeA{torre2011} for G-DINA model with binary responses, and extended to G-DINA with polytomous responses in \citeA{chen_torre2018}.  Specifically, given a latent profile $\bm{\alpha} = (\alpha_1, \dots, \alpha_K)$, we have
\begin{equation}
\label{GDINA}
    \text{log} \Big(\frac{\theta_{jrc}}{\theta_{j0c}} \Big) = b_{jr0} + \sum_{k = 1} ^{K} b_{jr1} q_{jk} \alpha_k + \sum_{k^{\prime} = k + 1 }^{K} \sum_{k = 1}^{K-1} b_{jrkk^{\prime}} (q_{jk} \alpha_k) (q_{j{k^{\prime}}} \alpha_{k^{\prime}}) + \cdots + b_{jr12\cdots K} \prod_{k = 1}^{K} q_{jk} \alpha_{k},
\end{equation}
where $b_{jr0}, b_{jr1}, \cdots$, $b_{jrK}, b_{jr12},$ $ \cdots, b_{jr(K-1)K}, \cdots, b_{jr12 \cdots K}$ are the coefficients in the generalized linear regression of the log-odds of conditional response probability on all latent attribute mastery situations, that is, all the subsets of $\{q_{j1}\alpha_1,\cdots,$ $q_{jK}\alpha_K \}$. Specifically, $b_{jr0}$ is the intercept of the log-odds; $b_{jr1}, \cdots$, $b_{jrK}$ are the main effects of attributes, representing the change of log-odds due to the mastery of the single attribute of $\alpha_1, \dots, \alpha_K$ respectively; $b_{jr12}, \cdots, b_{jr(K-1)K}, \cdots, b_{jr12 \cdots K}$ are the interaction effects of attributes, representing the change of log-odds due to the mastery of the combination of two or more attributes of $\alpha_1, \dots, \alpha_K$. 

For subjects with covariates values being zeros, the log-odds in G-DINA model \eqref{GDINA} is equivalent to general log-odds setting (\ref{log p}), which is the log-odds for RegLCMs and written as 
\begin{eqnarray}\label{z0}
\text{log}\Big(\frac{\theta_{jrc}}{\theta_{j0c}}\Big) &=& \gamma_{jrc} + \lambda_{1jr} z_{j1} + \cdots + \lambda_{qjr} z_{jq} = \gamma_{jrc} +  0 + \cdots +  0 =  \gamma_{jrc},  
\end{eqnarray}
which could be further expressed as 
\begin{equation}
	\theta_{jrc} = \frac {\text{exp}{(\gamma_{jrc} })}{ 1 + \sum_{s = 1}^{M_j - 1}\text{exp}( \gamma_{jsc}) }.\nonumber
\end{equation}
for $j = 1, \dots, J, r = 1, \dots, M_j-1$ and $c = 0, \dots, C-1$. We can show \eqref{GDINA} and \eqref{z0} are equivalent.
 Because in (\ref{GDINA}), the log-odds of conditional response probability are linear combinations of all the subsets of $\{q_{j1}\alpha_1,\ \cdots,\ q_{jK}\alpha_K \}$, and are dependent on latent profile $\bm{\alpha} = (\alpha_1, \cdots, \alpha_K)$ only, equivalently dependent on $c$ at $c = \bm{\alpha}^{T}\bm{v}$. 
 When covariates are zeros, the latent class category information is entirely captured by the intercept $\gamma_{jrc}$ in \eqref{z0}, implying that for given $j$ and $r$, each $\gamma_{jrc}$ is bijectively corresponding to $\bm{\alpha} \in \mathcal{A}$,  
which further implies that there exist a bijective linear correspondence between $\{b_{jr0}, b_{jr1}, \cdots, b_{jrK}, b_{jr12}, \cdots, b_{jr(K-1)K}, \cdots, b_{jr12 \cdots K}\}$  and $\{\gamma_{jrc}: c=0,\cdots, C-1\}$.

When covariates are involved in CDMs, we introduce the regression CDMs (RegCDMs) by    the following equations (\ref{RegCDMs log eta}) and (\ref{RegCDMs log p}) adapted from (\ref{log eta}) and (\ref{log p}), with the additional characteristics of CDMs that each latent membership $c$ is represented by a latent profile $\bm{\alpha}$. To make notations clear in this case, we denote the latent attributes of subject $i$ as $\bm{\alpha}_i = (\alpha_{i1}, \ldots, \alpha_{iK})$ for $i = 1, \dots, N$. And similarly as in RegLCMs, we use $\eta_c^i = P( \bm{\alpha}_i^{T}\bm{v} =c \mid \bm{x}_i, \bm{z}_i)$ to denote the latent class membership probability for subject $i$, and use $\theta_{jrc}^i = P( R_{ij} = r \mid  \bm{\alpha}_i^{T}\bm{v} = c ,\bm{x}_i, \bm{z}_i)$ to denote the conditional response probability for subject $i$ when these parameters are dependent on covariates.

 Assuming that the latent membership $c = 0$ denotes the latent profile that the subject $i$ does not master any of $K$ attributes, i.e. $\bm{\alpha}_i = \bm{0}_{K\times 1},$ we model
\begin{equation}
\label{RegCDMs log eta}
	\text{log}\Big( \frac{\eta_c^i}{\eta_0^i}\Big) = \beta_{0c} + \beta_{1c}x_{i1} + \cdots + \beta_{pc}x_{ip},
\end{equation}
for $ i = 1, \dots, N$ and $\bm{\alpha}_i^{T}\bm{v} =c$ with $\bm{\alpha}_i \in \{0,1\}^K \setminus \bm{0}_{K\times 1}$, and  
 \begin{equation}
\label{RegCDMs log p} 
\text{log}\Big(\frac{\theta_{jrc}^i}{\theta_{j0c}^i}\Big) = \gamma_{jrc} + \lambda_{1jr}z_{ij1} + \cdots + \lambda_{qjr}z_{ijq},
\end{equation}
for $i = 1, \dots, N$, $j = 1, \dots, J$, $r = 1, \dots, {M_j-1}$, and $\bm{\alpha}_i^{T}\bm{v} =c$ with $\bm{\alpha}_i \in \{0,1\}^K$. RegCDMs combine the regression setting on covariates from RegLCMs and the latent attribute representation from CDMs, which is to use binary latent profiles to represent latent classes. In addition,   Assumptions \ref{assumption1} and \ref{assumption2} in Section \ref{Covariate Effect Model} are also assumed for RegCDMs. 

\subsection{Identifiability Conditions in Existing Literature}
\label{local identifiability conditions in existing literature}
Before discussing our main results for the identifiability of the models introduced in Sections \ref{Covariate Effect Model} and   \ref{RegCDMs}, we give a review of the existing studies. 
The identifiability conditions for  latent class models have been extensively investigated in the existing literature. In particular,  \citeA{McHugh1956} studied the binary-response latent class models and proposed sufficient local identifiability conditions. Extending McHugh's work, \citeA{goodman1974} presented a fundamental method to determine the local identifiability of the polytomous-response latent class models, stating that if the Jacobian matrix formed by the derivatives of response probability vector with respect to parameters has full column rank, then the parameters are locally identifiable. This condition is intuitively straightforward but empirically nontrivial to apply. When the number of latent class $C$ or the number of possible responses to items $M_j$ increases, the dimension of the Jacobian matrix would increase at a fast rate. Moreover, this method could only guarantee the local identifiability for latent class models but leave the global identifiability undiscussed.

To study global identifiability, 
\citeA{kruskal1977} established algebraic results to ensure the uniqueness of factors in the decomposition of a three-way array. This work defined Kruskal rank which is analogous to the normal rank of a matrix. And it proved that if the Kruskal ranks of a triple product of matrices satisfy a certain arithmetic condition, the matrix decomposition will be unique. Based on Kruskal's theorems, \citeA{allman2009} extended the conditions to the decomposition into more than three variates and used them in the identifiability conditions for the latent class models with finite items. Besides, \citeA{allman2009} argued that even the parameters are not identifiable, the inference on parameters can be  valid empirically when the model is generically identifiable, that is, the parameters are identifiable except for a zero-measure set of parameters. The generic identifiability results allow us to circumvent the complex calculation on the column rank of the Jacobian matrix.

In the recent literature, the identifiability of the restricted latent class models, such as CDMs, has also been studied. Related identifiability results on  restricted models with binary responses were developed in \citeA{chen2015}, \citeA{xu_and_zhang_2016}, \citeA{xu2017}, \citeA{xu2018JASA},  \citeA{gu2019suf_nec_cond}, \citeA{gu2018partial},  etc. 
For the restricted latent class models with polytomous responses, \citeA{Culpepper2019}, \citeA{Fang2019},  \citeA{chen_culpepper2020},  and \citeA{gu2018partial} proposed the identifiability conditions dependent on the $Q$-matrix.

 The above research focuses on the identifiability of the general or restricted latent class models without covariates. For the identifiability of latent class models with covariates, \citeA{huang2004} generalized the result of \citeA{goodman1974} and derived local identifiability conditions for RegLCMs. 
 Under the setting of RegLCMs, denote $\mathcal{S}^{\prime}$ as the response pattern space   $\mathcal{S}$ with a reference pattern removed (e.g., $\mathbf 0$), so  the number of distinct response patterns in  $\mathcal{S}^{\prime}$  is then $S-1$.
 Define \begin{equation}
\bm{\Phi} = (\bm{\phi}_c; c = 0, \dots, C-1)_{(S-1) \times C}, \nonumber
\end{equation}
where each column $\bm{\phi}_c$ is of dimension $S-1$ in which each element corresponds to a response pattern $\bm{r}  = (r_{1}, \cdots, r_{J}) \in \mathcal{S}^{\prime}$  and is defined as
\begin{equation}
 \label{tau}
\phi_{\bm{r}c}= P(\bm{R} = \bm{r} \mid L = c,  \bm{z} = \bm{0}) = \prod_{j =1}^J \frac{ e^{\gamma_{jr_jc}} }{1 + \sum_{s = 1}^{M_j-1} e^{\gamma_{jsc}}},
 \end{equation}
where  
 $\gamma_{jr_jc}$ are defined as in (\ref{log p}) with $r= r_j$ and we set $\gamma_{j0c}=0$ for all $j = 1,\dots, j$ and $c = 0,\dots, C-1$. \citeA{huang2004} proposed that RegLCMs are locally identifiable at free parameters of $(\bm{\beta}, \bm{\gamma}, \bm{\lambda}) =\{\beta_{dc}, \gamma_{jrc}, \lambda_{tjr}:  j = 1, \ldots, J, r = 0, \ldots, M_j-1, c = 0, \ldots, C-1,d =0,\dots, p, t = 1,\dots, q\}$ if the following conditions are satisfied,
\begin{enumerate}[label=($A${\arabic*}),leftmargin=1.5cm]
     \item \label{A1} $\prod_{j=1}^{J} M_j-1 \geq C(\sum_{j=1}^{J} M_j -J) + C - 1 $;
    \item \label{A2}  Free parameters $\gamma_{jrc}, \lambda_{qjr}, \beta_{pc}$   and   covariate values $x_{ip}, z_{ijq}$ are all finite;
   
    \item \label{A3} The design matrix of the covariates
    \begin{equation}
  \bm{X} = \left(  \begin{array}{c}
 \bm{x}_{1}^{T} \\
\vdots \\
\bm{x}_{N}^{T} 
\end{array}\right)
= \left(  \begin{array}{cccc}
1 & x_{11} &\cdots& x_{1p}\\
\vdots & \vdots & \ddots & \vdots\\
1 & x_{N1} &\cdots & x_{Np}
\end{array}\right) \nonumber
\end{equation}
 and \begin{equation}
  \bm{Z_j} = \left(\begin{array}{cc}
 1& \bm{z}_{1j}^{T} \\
\vdots & \vdots \\
1& \bm{z}_{Nj}^{T}
\end{array}\right)=  \left(\begin{array}{cccc}
 1& z_{1j1} &\cdots& z_{1jq}\\
 \vdots & \vdots & \ddots & \vdots\\
1 & z_{Nj1} &\cdots & x_{Njq}
\end{array}\right), \quad j  = 1,\cdots, J \nonumber
\end{equation}
have full column rank;
\item \label{A4} $\bm{\phi_0}, \cdots, \bm{\phi_{C-1}}$ are linearly independent.
    \end{enumerate}

\begin{remark}
\label{remark F} As in \citeA{huang2004},
if we consider $F$ to be the number of pre-fixed conditional probabilities $\theta_{j rc} = 0 $ or $1$, then Condition \ref{A1} should be extended to $\prod_{j=1}^{J} M_j-1 \geq C(\sum_{j=1}^{J} M_j -J) + C - 1 - F$. For simplicity, we assume F = 0 throughout the paper.
\end{remark}
\begin{remark}
\label{A1A2 necessary}
	Condition \ref{A1} implies the number of independent response probabilities $$\prod_{j=1}^{J} M_j-1=\text{card}(\{P(R_1=r_1,\dots,R_J=r_J): r_j = 0, \dots, M_j-1, j = 1,\ldots,J \}),$$ exceeds the number of independent parameters in $(\bm{\eta}, \bm{\Theta})$,  $$C(\sum_{j=1}^{J} M_j -J) + C - 1 =\text{card}(\{\eta_c, \theta_{jrc} : j = 1,\dots, J, r = 1, \ldots, M_j -1, c = 0, \ldots, C-1\}).$$ Condition \ref{A1} is necessary, without which the observed response information may produce infinite parameter solutions and lead the model to be not identifiable. 
For technical rigorousness, Condition \ref{A2} as proposed in Huang and Bandeen-Roche (2004) specifies the model parameters and covariates $x_{ip}$, $z_{ijq}$ are finite. 
 In practice, the observed covariates are documented as finite values, and thus the finite condition on $x_{ip}$ and $z_{ijq}$ is automatically satisfied.
\end{remark}
For RegLCMs without covariates, which are  equivalent to RegLCMs with $\bm{x}_i = (1,\bm{0}_{1\times p})^{T}$ and $\bm{z}_i = \bm{0}_{J\times q}$, \citeA{huang2004} gave a reduced form of identifiability conditions. They claimed an equivalence between the full column rank condition on the Jacobian matrix and linear independence condition on the columns of marginal probability matrix  $\bm{\Psi}$ defined as 
 \begin{equation}
\bm{\Psi} = \left(\bm{\psi}_c; c = 0, \ldots, C-1\right)_{(S-1) \times C}, \nonumber
\end{equation}
where each column $\bm{\psi}_c$ is of dimension $S-1$ in which each element corresponds to a distinct response pattern $\bm{r}= (r_{1}, \cdots, r_{J}) \in \mathcal{S}^{\prime} $ and
\begin{equation}
\label{psi definition}
{\psi}_{\bm{r}c} = P(\bm{R} = \bm{r} \mid L = c) = \prod_{j=1}^{J} \prod_{r = 0}^{M_j-1} \theta_{jrc}^{\mathbb{I}\{r_{j} = r\}} = \prod_{j=1}^{J} \theta_{jr_jc}.
\end{equation}
Here for notational convenience, we let $\theta_{jr_jc}$ to denote $\theta_{jrc}$ defined in Section \ref{Covariate Effect Model} with $r = r_j$. 
Under the particular covariate latent class models with $\bm{x}_i = (1,\bm{0}_{1\times p})^{T}$ and $\bm{z}_i = \bm{0}_{J\times q}$, \citeA{huang2004} proposed that  $(\bm{\eta}, \bm{\Theta})= \{ \eta_c, \theta_{jrc}:  j = 1,\ldots, J, r = 0, \dots, M_j -1, c = 0, \ldots, C-1\}$ are locally identifiable if Condition \ref{A1} and the following conditions are satisfied:
	\begin{enumerate}[label=($A$\arabic*$^*$),leftmargin=1.5cm]
	 		 \setcounter{enumi}{1}
   	 \item \label{C2} For all free parameters, $\theta_{jrc} > 0$ and $\eta_c = P(L = c) > 0$;
    	\item \label{C3} $\bm{\psi}_0,\cdots,\bm{\psi}_{C-1}$ are linearly independent.
	\end{enumerate}

We  see that Conditions \ref{A1}--\ref{A3} and Condition \ref{C2}, are necessary  for the respective latent class models.  The necessity of Conditions \ref{A1} and \ref{A2} are discussed in Remark \ref{A1A2 necessary}. Condition \ref{C2} guarantees that the latent class membership probabilities and conditional response probabilities are non-zero. Condition \ref{A3} ensures $\bm\beta$, $\bm{\gamma}$ and $\bm{\lambda}$ are uniquely identifiable when $\bm{\eta}$ and $\bm{\Theta}$ are identifiable.  As for Condition \ref{C3}, it is related to the condition that the Jacobian matrix has full column rank. In the next section, we  show under the assumption that Conditions \ref{A1} and \ref{C2} hold, Condition \ref{C3} is necessary for the local identifiability of the special RegLCMs without covariates, but is actually not sufficient. Similarly, for RegLCMs with covariates, under Conditions \ref{A1}--\ref{A3}, Condition \ref{A4} is a necessary identifiability condition but  not a sufficient condition.

\section[]{Necessity but insufficiency of \citeA{huang2004}}
\label{Necessity and Sufficiency}

In this section, we show that the identifiability conditions in \citeA{huang2004} are not sufficient.
Following the discussion in Section \ref{local identifiability conditions in existing literature}, 
we first present the necessity of Condition \ref{A4} for RegLCMs and that of Condition \ref{C3} for RegLCMs without covariates, respectively. 

\newtheorem{proposition}{Proposition}
\begin{proposition}
\label{necessary c3}
For RegLCMs, Condition \ref{A4} is   necessary for the identifiability of $(\bm{\beta}, \bm{\gamma}, \bm{\lambda})$   under Conditions \ref{A1}--\ref{A3}. For RegLCMs without covariates, Condition \ref{C3} is   necessary for the identifiability of $(\bm{\eta}, \bm{\Theta})$  under Conditions \ref{A1} and \ref{C2}.
\end{proposition}

Despite the necessary results, we next show that satisfying Conditions \ref{A1}, \ref{C2} and \ref{C3} or satisfying Conditions \ref{A1}--\ref{A4} is not sufficient to guarantee the local identifiability of RegLCMs without or with covariates, respectively.
Our non-sufficient results are motivated by the existing works in the literature related to the identifiability of CDMs, which are a special family of RegLCMs as shown in   Section \ref{RegCDMs}. 
Specifically, we next present a proposition to show Conditions \ref{A1}, \ref{C2} and \ref{C3} are not sufficient for CDMs without covariates, and thus not sufficient for the identifiability of RegLCMs without covariates. Further, we show Conditions \ref{A1}--\ref{A4} are not sufficient for RegCDMs, and thus not sufficient for the identifiability of RegLCMs in general. 

\begin{proposition}
\label{non identifiability but full column rank}
Consider the setting of CDMs with polytomous responses. We assume Conditions \ref{A1}--\ref{A3} hold for     RegCDMs, and Conditions \ref{A1} and \ref{C2} hold for RegCDMs without covariates, i.e., CDMs.
If the following conditions hold:  
\begin{enumerate}[label=($P${\arabic*}),leftmargin=1.5cm]
\item \label{P1} Some latent attribute is required by only one item;
\item \label{P2} After rows permutation, the $Q$-matrix contains an identity matrix $\mathcal{I}_K$.
\end{enumerate}
Then we have
\begin{enumerate}[label=(\roman*)]
\item \label{prop2.i} For CDMs, the matrix $\bm{\Psi}$ in Condition \ref{C3} has full column rank but $(\bm{\eta}, \bm{\Theta})$  are not identifiable;
\item \label{prop2.i2}For  RegCDMs, the matrix $\bm{\Phi}$ in Condition \ref{A4} has full column rank but  $(\bm{\beta}, \bm{\gamma}, \bm{\lambda})$ are not identifiable.
\end{enumerate}
\end{proposition}

According to Proposition \ref{non identifiability but full column rank},  the $Q$-matrix as shown in the following form satisfies Conditions \ref{P1} and \ref{P2},
\begin{equation}
Q = \left(\begin{array}{ccc}
& \mathcal{I}_K~~~ & \\ \cline{1-3} 
 \begin{array}{c}
0 \\
\vdots  \\
0
\end{array} &  {Q^{*}}& \end{array}\right). \nonumber
\end{equation}
The above $Q$-matrix is complete as the top $K \times K$ block is an identity matrix $\mathcal{I}_K$. From the $(K + 1)$th row to the $J$th row, the entries in the first column are $\bm{0}_{J-K}$, and the entries in the remaining columns are denoted as a submatrix $Q^{*}$.
The first result \ref{prop2.i} in Proposition \ref{non identifiability but full column rank} is derived by extending a similar  conclusion for  CDMs with binary responses in \citeA{gu2018partial}
  to CDMs with polytomous responses. 
  With a complete $Q$-matrix, the matrix $\bm{\Psi}$ in Condition \ref{C3} can be shown to have full column rank, or equivalently, $\bm{\psi}_0,\cdots,\bm{\psi}_{C-1}$ are linearly independent. 
And further, we can show that for RegCDMs, the matrix $\bm{\Phi}$ in Condition \ref{A4} has full column rank, that is, $\bm{\phi}_0,\cdots,\bm{\phi}_{C-1}$ are linearly independent.


With Proposition \ref{non identifiability but full column rank}, we see that   RegLCMs without covariates may not be identifiable when Conditions \ref{A1}, \ref{C2} and \ref{C3} are satisfied. 
Specifically, consider CDMs without covariates, 
 given Conditions \ref{P1}--\ref{P2} of Proposition \ref{non identifiability but full column rank} are satisfied, Condition  \ref{C3} will be true since $\bm{\psi}_0,\cdots,\bm{\psi}_{C-1}$ are linearly independent. However, Proposition \ref{non identifiability but full column rank}\ref{prop2.i} shows that such CDMs are not identifiable. Therefore Conditions \ref{A1}, \ref{C2} and \ref{C3} are not sufficient for the identifiability of CDMs.

Similarly, RegLCMs may not be identifiable provided that Conditions \ref{A1}--\ref{A4} hold. For RegCDMs, 
 given Conditions \ref{A1}--\ref{A3} and Conditions \ref{P1}--\ref{P2} of Proposition \ref{non identifiability but full column rank} are met, Condition  \ref{A4} will be true since $\bm{\phi_0}, \cdots, \bm{\phi_{C-1}}$ are linearly independent, but such RegCDMs are not identifiable according to Proposition \ref{non identifiability but full column rank}\ref{prop2.i2}. Therefore Conditions \ref{A1}--\ref{A4} are not sufficient for the identifiability of RegCDMs.

\section{Sufficient and Practical Identifiability Conditions}
\label{proposed identifiability conditions}

As shown in Section \ref{Necessity and Sufficiency}, 
Conditions \ref{A1}--\ref{A4} are necessary but not sufficient for the identifiability of RegLCMs. 
To address the issue, this section provides  sufficient conditions to determine the identifiability of RegLCMs. In addition,  we also establish sufficient identifiability conditions for RegCDMs, which   are of great importance in cognitive diagnosis.

For completeness, we first review the fundamental method to check the local identifiability before discussing the strict and generic identifiability. In Section \ref{local identifiability conditions in existing literature}, we have introduced the results of the local identifiability conditions proposed by \citeA{goodman1974}. The conditions can be generalized to finite many items and under the setting of RegLCMs.

We first consider RegLCMs without covariates.  The definitions of conditional response probabilities $\theta_{jrc}$ follow from Section \ref{Covariate Effect Model}.
For $\bm{r} =(r_{1}, \cdots, r_{J}) \in \mathcal{S}$,
recall that we denote the response probability as 
\begin{eqnarray}
	P(\bm{R} = \bm{r}) = \sum_{c = 0}^{C-1} \eta_c P(\bm{R} = \bm{r} \mid L = c)   
	=  \sum_{c = 0}^{C-1} \eta_c  \prod_{j=1}^{J}\theta_{jr_jc}. \nonumber
\end{eqnarray}
The local identifiability condition proposed by Goodman is associated with the Jacobian matrix 
\begin{equation}
\mathbf{J}=\left(\bm{J}_{\eta_1}, \cdots,\bm{J}_{\eta_{C-1}}, \bm{J}_{\theta_{110}}, \cdots, \bm{J}_{\theta_{1(M_1-1)0}},\cdots,\bm{J}_{\theta_{J1(C-1)}}, \cdots, \bm{J}_{\theta_{J(M_J-1)(C-1)}} \right). \nonumber
\end{equation}
The row dimension of $\mathbf{J}$ is $S-1$ and the column dimension is $ C(\sum_{j=1}^{J} M_j -J) + C - 1  $, where each row index corresponds to one response probability $P(\bm{R} = \bm{r})$ for $\bm{r} \in \mathcal{S}^{\prime} $  and each column index corresponds to one free parameter from  $ \{ {\eta_1}, \cdots,{\eta_{C-1}}, {\theta_{110}}, \cdots, {\theta_{1(M_1-1)0}},\cdots,$ ${\theta_{J1(C-1)}}, \cdots, {\theta_{J(M_J-1)(C-1)}} \}$. For $c = 1,\dots,C-1$, $\bm{J}_{\eta_c}$ is a vector of dimension $S - 1$. Each entry is a partial derivative of the response probability $P(\bm{R} = \bm{r})$  with respect to $\eta_c$ at true value of $\eta_c$, which is computed to be
\begin{equation}
	\frac{\partial{P(\bm{R} = \bm{r})}}{\partial{\eta_c}} = \prod_{j=1}^{J} \theta_{jr_jc}- \prod_{j=1}^{J} \theta_{jr_j0}. \nonumber
\end{equation}
And for $j = 1, \dots, J$, $r = 1, \dots, M_j-1$ and $c = 0, \dots, C-1$, $\bm{J}_{\theta_{jrc}}$ is a vector of dimension $S - 1$. Each entry is a partial derivative of the response probability $P(\bm{R} = \bm{r})$   with respect to $\theta_{jrc}$ at true value of $\theta_{jrc}$, which is computed to be
\begin{equation}
\label{jacobian theta}
	\frac{\partial{P(\bm{R} = \bm{r})}}{\partial{\theta_{jrc}}} =
  \begin{cases}
     \eta_c \prod_{d \neq j} \theta_{dr_dc}, & \text{if $r_{j} = r$}; \\
     - \eta_c \prod_{d \neq j} \theta_{dr_dc}, & \text{if $r_{j} = 0$}; \\
  0, & \text{otherwise}.
  \end{cases} \nonumber
\end{equation}


\smallskip
\newtheorem{theorem}{Theorem}
\begin{theorem}[Local Identifiability for LCMs and CDMs]
\label{local RegLCMs}
 Consider RegLCMs without covariates or CDMs. Under Conditions \ref{A1} and \ref{C2}, $(\bm{\eta}, \bm{\Theta})$ are locally identifiable if and only if the following condition holds.
\begin{enumerate}[label=($A$\arabic*$^{**}$),leftmargin=1.5cm]
 		 \setcounter{enumi}{2}
 \item The Jacobian matrix $\mathbf{J}$ formed above  has full column rank.
   \end{enumerate}
 \end{theorem}
 
 To better present the following local identifiability theorem for RegLCMs and RegCDMs, we consider a ``hypothetical"  subject with all covariates being zeros, that is, $\bm{x} = (1,\bm{0}_{1\times p})^{T}$ and $\bm{z} = \bm{0}_{J\times q}$. Denote the parameters of this particular subject to be $\bm{\eta}^0$ and $\bm{\Theta}^0$. The Jacobian matrix $\mathbf{J}^0$ formed by the derivatives of conditional response probabilities with respect to parameters  $\bm{\eta}^0$ and $\bm{\Theta}^0$ is equivalent to the computation of Jacobian matrix $\mathbf{J}$ of general restricted latent class models shown in Theorem \ref{local RegLCMs}. Next, we present a theorem to associate the $\mathbf{J}^0$ with the local identifiability of   $(\bm{\beta}, \bm{\gamma}, \bm{\lambda})$.
 
 \begin{theorem}[Local Identifiability for RegLCMs and RegCDMs]
 \label{local RegCDMs}
 Consider RegLCMs or  RegCDMs. Under Conditions \ref{A1}--\ref{A3}, $(\bm{\beta}, \bm{\gamma}, \bm{\lambda})$ are locally identifiable if and only if the following condition holds.
\begin{enumerate}[label=($A$\arabic*$^{\prime}$),leftmargin=1.5cm]
 		 \setcounter{enumi}{3}
 \item The Jacobian matrix $\mathbf{J}^0$ formed from the hypothetical subject with covariates being zeros has full column rank.
   \end{enumerate}
 \end{theorem}

 Theorems \ref{local RegLCMs} and \ref{local RegCDMs} are intuitively straightforward but nontrivial to apply in practice. When the number of latent classes $C$ and the number of item responses $M_j$ increase, the dimension of the Jacobian matrix would increase, making it challenging to compute the rank of the Jacobian matrix.

Moreover, the conditions introduced in Theorems \ref{local RegLCMs} and \ref{local RegCDMs} only guarantee the local identifiability, while the global identifiability is not discussed. To ensure the sufficiency for global strict identifiability, we combine Goodman's idea with the algebraic results from Kruskal Theorem to establish our conditions. Recall that 
$
\bm{\Phi} = (\bm{\phi}_c; c = 0, \ldots, C-1) \nonumber
$
defined in (\ref{tau}) is a matrix of dimension $(S-1) \times C$.
And $\bm{\phi}_c$ is a  vector where each element corresponds to one response pattern and is denoted as
$
\phi_{\bm{r}c}= P(\bm{R} = \bm{r} \mid L = c,  \bm{z} = \bm{0}). 
$
To apply Kruskal Theorem and to establish the strict identifiability conditions, 
 we consider a three-way decomposition of $\bm{\Phi}$ and propose the linear independence condition regarding the decomposed matrices instead of $\bm{\Phi}$. We divide the total of $J$ items of $\bm{\Phi}$ into three mutually exclusive item sets $\mathcal{J}_1, \mathcal{J}_2$ and $\mathcal{J}_3$ containing $J_1, J_2$ and $J_3$ items respectively, with $J_1 + J_2 + J_3 = J$. For $t = 1,2$ and $3$, each set $\mathcal{J}_t$ can be viewed as   one polytomous variable $T_t$ taking on values in $\{1, \cdots,\kappa_t\}$ with cardinality $\kappa_t = \prod_{j \in \mathcal{J}_t } M_j$ to be the number of response patterns for this set. And each variable $T_t$ is used to construct a $\kappa_t \times C$ submatrix $\bm{\Phi}_t$, where its row indices arise from the response patterns corresponding to $T_t$. The linear independence condition is then regarding to the Kruskal ranks of $\bm{\Phi}_t$ rather than normal column rank of $\bm{\Phi}$, 
 where for any matrix $\bm{\Phi}_t$, its Kruskal rank $I_t$ is the smallest number of columns of $\bm{\Phi}_t$ that are linearly dependent. 

\begin{theorem}[Strict Identifiability for RegLCMs]
 \label{covariate strict identifiability regLCMs}
 Continue with the notation definitions in Section \ref{local identifiability conditions in existing literature}. For RegLCMs, under Conditions \ref{A1}--\ref{A3} and the following condition, $(\bm{\beta}, \bm{\gamma}, \bm{\lambda})$ are strictly identifiable. 
 \begin{enumerate}[label=($C$\arabic*),leftmargin=1.5cm]
     		 		 \setcounter{enumi}{3} 
  \item \label{C4NEW}The matrix $\bm{\Phi}$ can be decomposed into $\bm{\Phi}_1$, $\bm{\Phi}_2$ and $\bm{\Phi}_3$ with Kruskal ranks $I_1$, $I_2$ and $I_3$ satisfying $I_1 + I_2 + I_3 \geq 2C + 2$.

 \end{enumerate}
\end{theorem}

Theorem \ref{covariate strict identifiability regLCMs} is sufficient to guarantee the strict identifiability for RegLCMs, including RegCDMs. Compared with the local identifiability conditions in \citeA{huang2004}, Theorem \ref{covariate strict identifiability regLCMs} keeps Conditions \ref{A1}--\ref{A3} and replaces Condition \ref{A4} concerning the column rank of $\bm{\Phi}$ with a stronger Condition \ref{C4NEW} concerning the Kruskal ranks of the decomposed matrices from $\bm{\Phi}$. This condition is based on the algebraic result in \citeA{kruskal1977}.
We next present   identifiability   conditions tailored to    RegCDMs.

\begin{proposition} [Strict Identifiability for RegCDMs]
 \label{covariate strict identifiability CDMs}
For RegCDMs with polytomous responses, under Conditions \ref{A1}--\ref{A3} and   the following condition,    $(\bm{\beta}, \bm{\gamma}, \bm{\lambda})$ are strictly identifiable.
 		\begin{enumerate}[label=($C$\arabic*$^*$),leftmargin=1.5cm]
 		 		 \setcounter{enumi}{3}
  			\item \label{C4star}After rows permutation,  $Q$-matrix takes the form $Q = ( \mathcal{I}_K, \mathcal{I}_K, Q^{*}  )^{T} $ containing two identity matrices $\mathcal{I}_K$ and one  submatrix $Q_{(J-2K) \times K}^{*}$. And  for any different latent classes $c$ and $c^{\prime}$, there exist at least one item $j > 2K$ such that $({\theta}_{j0c}, \cdots, \theta_{j(M_j-1)c})^{T} \neq ({\theta}_{j0c^{\prime}}, \cdots, \theta_{j(M_j-1)c^{\prime}})^{T}$.  
 		\end{enumerate}

 \end{proposition}

It has been established that Condition \ref{C4star} itself is a sufficient   condition for the identifiability of general restricted latent class models with binary responses \cite{xu2017}. In addition, \citeA{xu2018JASA} showed that the $Q$-matrix is also identifiable under Condition \ref{C4star}. This condition is further extended to the restricted latent class models with polytomous responses in \citeA{Culpepper2019}. Compared to the previous literature, the major contribution of Proposition \ref{covariate strict identifiability CDMs} is to extend this constraint to the polytomous-response RegCDMs that the $Q$-matrix contains two identity matrices and the conditional response probability $({\theta}_{j0c}, \cdots, \theta_{j(M_j-1)c})^{T}$  is distinct among different latent classes.

In practice, the theoretical results in Theorem \ref{covariate strict identifiability regLCMs} and Proposition \ref{covariate strict identifiability CDMs} may need further adjustments to accommodate the empirical needs. As previously discussed, generic identifiability is commonly used in practice as it guarantees the identifiability of most parameters other than a measure-zero set of parameters \cite{allman2009}. The following theorem and proposition will provide us with an easy way to determine the generic identifiability of RegLCMs and RegCDMs.

 \begin{theorem} [Generic Identifiability for RegLCMs]
 \label{generic RegLCMs}
 For RegLCMs, under Conditions \ref{A1}--\ref{A3} and the following condition, $(\bm{\beta}, \bm{\gamma}, \bm{\lambda})$ are generically identifiable.
 		\begin{enumerate}[label=($C$\arabic*$^{\prime}$),leftmargin=1.5cm]
 		 \setcounter{enumi}{3}
 			\item \label{C4prime} The matrix $\bm{\Phi}$ can be decomposed into $\bm{\Phi}_1$, $\bm{\Phi}_2$ and $\bm{\Phi}_3$ with row dimensions $\kappa_1$, $\kappa_2$ and $\kappa_3$ satisfying $ \min\{C, \kappa_1 \} + \min\{C, \kappa_2 \} + \min\{C, \kappa_3 \} \geq 2C + 2$. 			
 		\end{enumerate}	
 \end{theorem}
\begin{remark}
	Under the special case that the number of possible responses to each item are identical, $M_1 = \dots = M_J$, we have a reduced form of Condition \ref{C4prime} in Theorem \ref{generic CDMs}. This finding is based on Corollary 5 and its related discussions from \citeA{allman2009}. They show that for these special cases, the decomposition can be carefully chosen to maximize  $ \min\{C, \kappa_1 \} + \min\{C, \kappa_2 \} + \min\{C, \kappa_3 \}$, which results in a simpler form of identifiability condition.
	
	 Consider RegLCMs with binary responses $M_j = 2$ for $j = 1,\dots, J$, we have $(\bm{\beta}, \bm{\gamma}, \bm{\lambda})$ to be generically identifiable if we replace 	Condition \ref{C4prime} with the condition $J \geq 2\left[\log _{2} C\right\rceil+1$. More generally, for the RegLCMs with $M_j = M$ for $j = 1,\dots, J$, we have $(\bm{\beta}, \bm{\gamma}, \bm{\lambda})$ to be generically identifiable if Condition \ref{C4prime} is replaced with the condition $J \geq 2\left[\log _{M} C\right\rceil+1$. For these special models, the reduced conditions provide researchers with simpler ways to determine the generic identifiability compared with Condition \ref{C4prime} as they only concern the number of items $J$ and the number of latent classes $C$.
	\end{remark}

Compared with the strict identifiability conditions in Theorem \ref{covariate strict identifiability regLCMs}, Theorem \ref{generic RegLCMs} makes it more practical to check the identifiability of RegLCMs as the variables in Condition \ref{C4prime} are row dimensions rather than the Kruskal ranks of the decomposed matrices. But Theorem \ref{generic RegLCMs} does not apply to all latent class models. For instance, the parameter space of restricted latent class models may lie in the nonidentifiable measure-zero set from the parameter space of general latent class models. Therefore, Theorem \ref{generic RegLCMs} does not apply to restricted latent class models with covariates such as RegCDMs. To address this issue, Proposition \ref{generic CDMs} is established to determine the generic identifiability  for RegCDMs with polytomous responses.

 \begin{proposition} [Generic Identifiability for RegCDMs]
 \label{generic CDMs}
 For RegCDMs with polytomous responses, under Conditions \ref{A1}--\ref{A3} and the following condition,  $(\bm{\beta}, \bm{\gamma}, \bm{\lambda})$ are generically identifiable.

 		\begin{enumerate}[label=($C$\arabic*$^{\prime\prime}$),leftmargin=1.5cm]
 		 	\setcounter{enumi}{3}
 				\item \label{C4prime2} After rows permutation,  $Q$-matrix takes the form $Q= ( Q_1, Q_2, Q^{*})^{T} $ containing one submatrix $Q_{(J-2K) \times K}^{*}$ in which each attribute is required by at least one item, and two submatrices $Q_1$ and $Q_2$ in the following form,
 			\begin{equation}
 			\label{Q_i}
 				Q_i = \left(\begin{array}{cccc}
					1& * & \cdots & * \\  
					*& 1 &  \cdots &* \\
					\vdots & \vdots & \ddots & \vdots \\
					* & * & \cdots & 1 \\
 					\end{array} \right),\quad i =1,2
 			\end{equation}
 			where  $``*"$ indicates the entry is either 1 or 0. 
 		\end{enumerate}
 \end{proposition}

Condition \ref{C4prime2} was first proposed by \citeA{gu2018partial} to determine the generic identifiability of CDMs. For RegCDMs, Proposition \ref{generic CDMs} gives more flexible conditions than Proposition \ref{covariate strict identifiability CDMs} as Condition \ref{C4prime2} puts less constraints on the $Q$-matrix than Condition \ref{C4star} does. Condition \ref{C4star} requires the $Q$-matrix to contain two identity submatrices, whereas in the $Q$-matrix form required by \ref{C4prime2}, the two identity matrices are replaced by two matrices as shown in \eqref{Q_i}, which allows more flexibility on the off-diagonal entries. Under this new condition, the parameters may not be strictly identifiable but are identifiable in the generic sense. 
 

 Proposition \ref{generic CDMs}  provides sufficient conditions to guarantee the generic identifiability of RegCDMs. 
Under certain special cases,
  we can  show that those conditions are also necessary.
 Next, we introduce a particular example where the  conditions in Proposition \ref{generic CDMs} are not only sufficient, but also necessary for the generic identifiability of the parameters $(\bm{\beta}, \bm{\gamma}, \bm{\lambda})$.

\begin{example}

\label{suff nec generic CDMs}
	Consider a special RegCDM with binary responses and two latent attributes, i.e. $K=2$ and $M_j = 2$. Under Conditions \ref{A1}--\ref{A3}, Condition \ref{C4prime2} in Proposition \ref{generic CDMs} is necessary and sufficient for the generic identifiability of  $(\bm{\beta}, \bm{\gamma}, \bm{\lambda})$. For instance, after rows permutation, the $Q$-matrix takes the following form
\begin{equation}
	\label{two attribute Q}
		Q = \left(\begin{array}{cccc}
					1& *  \\  
					*& 1  \\
					1& *  \\  
					*& 1  \\
					\cline{1-3} 
					\multicolumn{2}{c}{Q^{\prime}}\\
 					\end{array} \right),
\end{equation}
where $``*"$ is either zero or one and $Q^{\prime}$ is a matrix with at least one entry to be 1 in each column. 
Proposition 3 in \citeA{gu2018sufficient} shows that Condition \ref{C4prime2} is necessary and sufficient condition for generic identifiability for $Q$-matrix, $\bm{\Theta}$ and $\bm{\eta}$. Hence for RegCDMs, we have  $(\bm{\eta}^i$, $\bm{\Theta}^i)$ identifiable. As for the identifiability of $(\bm{\beta}, \bm{\gamma}, \bm{\lambda})$ in RegCDMs, under Condition \ref{A3} that $\bm{X}$ and $\bm{Z_j}$'s have full column rank, $(\bm{\eta}^i$, $\bm{\Theta}^i)$ are identifiable if and only if  $(\bm{\beta}, \bm{\gamma}, \bm{\lambda})$ are identifiable, which can be seen from \textit{Steps 2--3} of the \textit{Proof of Theorem \ref{local RegCDMs}} in Supplementary Material. Therefore, $(\bm{\beta}, \bm{\gamma}, \bm{\lambda})$ are identifiable for the considered RegCDMs with two attributes if and only if Condition \ref{C4prime2} in Proposition \ref{generic CDMs} holds.
\end{example}

\section{Data Example}
\label{Data Example}

In this section, we  use a real data set to demonstrate an application of the proposed identifiability conditions in educational assessments. Trends in Mathematics and Science Study (TIMSS) is an international and large-scale assessment to evaluate the mathematics skills and science knowledge of students in different grades. We consider a TIMSS 2007 4th Grade dataset, which was  studied in \citeA{park2014} and is accessible from the R  package ``CDM" \cite{jsscdm, CDM_7.5-15}. The dataset contains $N = 698$ Austrian 4th grade students' binary responses ($M_j = 2$) to $J = 25$ items together with their gender information. The gender  is denoted as a binary variable with $g_i = 1$ for female students and $g_i =0$ for male students. 

We  model the TIMSS 2007 dataset using RegCDMs and study their identifiability. We consider gender $g_i$ as covariates  with $\bm{x}_i = (1, g_i)^T$ and $\bm{z}_{ij} = (g_i)$ for $i = 1,\dots, N$ and $j = 1,\dots, J$, under the assumption that both $\bm{\eta}$ and $\bm{\Theta}$ can be associated with the gender. Following \citeA{park2014}, the test assesses $K= 7$ latent attributes in the domains of $(\alpha_1)$ Whole numbers; $(\alpha_2)$ Fractions and Decimals; $(\alpha_3)$ Number Sentences, Patterns, \& Relationships; $(\alpha_4)$ Lines and Angles; $(\alpha_5)$ Two- and Three-Dimensional Shapes; $(\alpha_6)$ Location and Movement; $(\alpha_7)$ Reading, Interpreting, Organizing, \& Representing. As shown in \citeA{park2014}, the seven latent attributes can be further aggregated into $K^{\prime} = 3$ general domains: $(\alpha_1^{\prime})$ Number; $(\alpha_2^{\prime})$ Geometric Shapes and Measures; $(\alpha_3^{\prime})$ Data Display.

We first show that the RegCDM with $K =7$ attributes is generically identifiable by Proposition \ref{generic CDMs}. As there are $C=2^7 = 128$ latent classes, Condition \ref{A1} holds as $\prod_{j=1}^{J} M_j-1 - C(\sum_{j=1}^{J} M_j -J) - C + 1 = 2^{25} - 2^7\times 25 - 2^7  >0$. Condition \ref{A2} holds as the binary covariates are finite and coefficient parameters are free since we have no constraint on coefficients. Condition \ref{A3} holds as the design matrices  
    \begin{equation}
  \bm{X}=  \bm{Z_j} = \left(  \begin{array}{cc}
1 & g_1\\
1 & g_2\\
\vdots & \vdots\\
1 & g_N 
\end{array}\right) =\left(  \begin{array}{cc}
1 & 0\\
1 & 0\\
\vdots & \vdots\\
1 & 1
\end{array}\right), \quad \text{for } j = 1,\dots, J, \nonumber
 \end{equation}
  have full column rank given the sample has both female and male students. Lastly for Condition \ref{C4prime2}, the $Q$-matrix after rows permutation from \citeA{park2014} is presented in Table \ref{q7matrix}. The $Q$-matrix implies that Condition \ref{C4prime2} holds as the matrices $Q_1$ and $Q_2$ have diagonal entries to be ones and each column of the sub-matrix $Q^*$ contains the value one for at least once.
\begin{table}[h!]
\centering
\begin{tabular}{c|c|ccccccc}
\hline 
& Item No.  & $\alpha_1$  & $\alpha_2$ & $\alpha_3$& $\alpha_4$& $\alpha_5$& $\alpha_6$& $\alpha_7$ \\
   \hline
\multirow{ 7}{*}{$Q_1$}& 1 & 1  & 0    & 0   & 0   & 0   & 0   & 0 \\
&3 & 1  & 1    & 0   & 0   & 0   & 0   & 0  \\

&5  & 1 & 0 & 1 & 0 & 0 & 0 & 0  \\
&10 &  0 & 0 & 0 & 1 & 1 & 0 & 0 \\
&9 & 0  & 0    & 0   & 0   & 1   & 0   & 0 \\
&6 & 0  & 0    & 0   & 0   & 1   & 1 & 0 \\

&12 &  1 & 0 & 0 & 0 & 0 & 0 & 1 \\
\hdashline
\multirow{ 7}{*}{$ Q_2$}& 15 & 1  & 0    & 0   & 0   & 0   & 0   & 0 \\
&4 & 1  & 1    & 0   & 0   & 0   & 0   & 0  \\

&17  & 1 & 0 & 1 & 0 & 0 & 0 & 0  \\
& 11 &  1 & 0 & 0 & 1 & 0 & 0 & 0 \\
& 24 & 0  & 0    & 0   & 0   & 1   & 0   & 0 \\
&22 & 0  & 0    & 0   & 0   & 1   & 1 & 0 \\

& 13 &  1 & 0 & 0 & 0 & 0 & 0 & 1 \\
\hdashline
\multirow{ 7}{*}{$Q^*$} 
&2  & 0 & 1    & 0   & 0   & 0   & 0   & 0 \\
&8  & 1  & 0    & 0   & 0   & 1   & 0   & 0\\
&7 &  0 & 0 & 0 & 1 & 1 & 1 & 0 \\
&14 & 1  & 1    & 0   & 0   & 0   & 0   & 1 \\
& 16, 23 & 1  & 0    & 0   & 0   & 0   & 0   & 0 \\
&18, 20  & 1 & 0 & 1 & 0 & 0 & 0 & 0  \\
& 19, 25 &  1 & 0 & 0 & 0 & 0 & 0 & 1 \\
& 21  & 1 & 0 & 1 & 0 & 0 & 0 & 0 \\
\hline
\end{tabular}
\caption{The $Q$-matrix for TIMSS 2007 Data at $K = 7$.  }
\label{q7matrix}
\end{table}
According to Proposition \ref{generic CDMs}, the RegCDM is generically identifiable. However, the $Q$-matrix is not complete, so the RegCDM is not strictly identifiable.

We next show that the RegCDM  with $K^{\prime} =3$ attributes is generically identifiable as well by Proposition \ref{generic CDMs}. As there are $C=2^3 = 8$ latent classes, Condition \ref{A1} holds because $\prod_{j=1}^{J} M_j-1 - C(\sum_{j=1}^{J} M_j -J) - C + 1 = 2^{25} - 2^3\times 25 - 2^3  >0$. As the items, the students' responses, and the covariates are unchanged, we have Conditions \ref{A2}--\ref{A3} hold by the same arguments as in the RegCDM with $K= 7$. In assessing the three general attributes, the $Q$-matrix used in \citeA{park2014} is given in Table \ref{q3matrix} after rows permutation.  This $Q$-matrix contains $Q_1$ and $Q_2$ with diagonal entries to be ones and the sub-matrix $Q^*$ with each attribute column containing the value one for at least one entry. Therefore, Condition \ref{C4prime2} holds and Proposition \ref{generic CDMs} shows the RegCDM  with $K^{\prime} = 3$ is generically identifiable. However, the $Q$-matrix does not contain an identity matrix as the $\alpha_3^{\prime}$ is not singularly required by any item. So the RegCDM with $K^{\prime} = 3$ is not strictly identifiable.

\begin{table}[h!]
\centering
\begin{tabular}{c|c|ccc}
\hline 
& Item No.  & $\alpha_1^{\prime}$  & $\alpha_2^{\prime}$ & $\alpha_3^{\prime}$ \\
   \hline
\multirow{ 3}{*}{$Q_1$}& 1 & 1  & 0    & 0   \\
&6 & 0  & 1    & 0    \\
&12 &  1 & 0 & 1  \\
\hdashline
\multirow{ 3}{*}{$ Q_2$}& 2 & 1  & 0    & 0   \\
& 7 & 0  & 1    & 0    \\
&13 &  1 & 0 & 1  \\
\hdashline
\multirow{ 4}{*}{$Q^*$} 
& 3--5, 15--18, 21, 23 & 1  & 0    & 0   \\
&9, 10, 22, 24 & 0  & 1    & 0    \\
&14, 19, 20, 25 &  1 & 0 & 1  \\
& 8, 11 & 1 & 1 & 0 \\
\hline
\end{tabular}
\caption{The $Q$-matrix for TIMSS 2007 Data at $K = 3$.  }
\label{q3matrix}
\end{table}


\section{Discussion}
\label{Conclusion}
This paper studies latent class models with covariates, in particular, RegLCMs. Under the setup of RegLCMs and its special family RegCDMs, we focus on the identifiability conditions for the coefficient parameters of the covariates. We show that \citeA{huang2004} presented necessary but not sufficient conditions for the local identifiability of RegLCMs. Then we establish conditions for the local and global identifiability of RegLCMs and RegCDMs.

The classical and fundamental method for local identifiability is based on Goodman's results, which is to ensure the full column rank of the Jacobian matrix formed by the derivatives of general response probabilities with respect to parameters. We propose sufficient and practical conditions based on \citeA{huang2004} to replace the previous linear independence condition on the marginal probability matrix with the linear independence condition concerning three decomposed probability matrices. Noticing the empirical convenience of the generic identifiability, we present specific conditions to ensure the generic identifiability as well. The conditions for generic identifiability involve more accessible variables from decomposed submatrices. In addition to the global identifiability of general RegLCMs, the conditions for the global identifiability of RegCDMs   are dependent on the $Q$-matrix, and these conditions are extended from the binary-response CDMs to the polytomous-response CDMs.

Regarding the consistency of estimation, \citeA{gu2018partial} proved that for  general restricted latent class models, the latent class membership probability and conditional response probability can be consistently estimated with maximum likelihood estimators. The estimation consistency is retained for the parameters in RegLCMs because the parameters are linearly related with the log-odds and the design matrices of covariates have full column ranks. The proposed conditions are sufficient and practical, but may not be necessary in strict identifiability cases. For generic identifiability, we discuss the sufficient and necessary conditions for the binary-response CDMs with binary attributes in Example \ref{suff nec generic CDMs}, except which the necessary side of identifiability conditions is still under research. For future works, we plan to investigate the sufficient and necessary conditions for the identifiability of latent class models with covariates.

\section*{Acknowledgments}
The authors are grateful to the Editor-in-Chief Professor Matthias von Davier, an Associate Editor, and a referee for their valuable comments and suggestions. This research is partially supported by NSF CAREER SES-1846747 and Institute of Education Sciences R305D200015.

\bibliographystyle{apacite}
\bibliography{ref}

\newpage

\begin{center}
  \textbf{\large Supplemental Material to ``Identifiability of Latent Class Models with Covariates"}\\[.5cm]
    \end{center}
    
This supplementary material contains two sections. Section \ref{Proofs of Propositions and Theorems} provides the proofs of propositions and theorems from Section \ref{Necessity and Sufficiency} and Section \ref{proposed identifiability conditions} of the main article.  Section \ref{Proofs of Lemmas} gives the proofs of lemmas introduced in Section \ref{Proofs of Propositions and Theorems}.

\appendix
\addappheadtotoc

\label{proof appendix}

\section{Proofs of Propositions and Theorems}
\label{Proofs of Propositions and Theorems}
In this section, we first introduce a lemma motivated from Proposition 3 in \citeA{huang2004}, which is an important tool in later proofs to associate the identifiability of parameters $(\bm{\beta}, \bm{\gamma}, \bm{\lambda})$ with the identifiability of $(\bm{\eta}^i, \bm{\Theta}^i)=\{\eta_c^i, \theta_{jrc}^i:  j = 1, \dots,J, r = 0,\dots, M_j -1, c = 0,\dots, C-1 \}$, for $i = 1, \ldots, N$. 


\begin{lemma}
\label{tau and tau 0}
For any subject $i = 1, \ldots, N$, we define transformed variables $(\bm{\epsilon}^i,\bm{\omega}^i) =\{\epsilon_c^i, \omega_{jrc}^i :  j = 1,\ldots, J, r = 0, \dots, M_j -1, c = 0, \ldots, C-1  \}$ such that $(\bm{\eta}^i, \bm{\Theta}^i)$ and $(\bm{\epsilon}^i$, $\bm{\omega}^i) $ are related through the following equations, 
\begin{eqnarray}
\eta_c^i = \frac{\exp(\epsilon_c^i)}{1 + \sum_{s=1}^{C-1}\exp(\epsilon_s^i) }, \quad  c &=& 0, \ldots, C-1; \nonumber \\
	\theta_{jrc}^i = \frac {\exp(\omega_{jrc}^i)}{ 1 + \sum_{s = 1}^{M_j - 1} \exp(\omega_{jsc}^i)}, \quad  j &=&1, \ldots, J;\nonumber \\
	r&=& 0, \ldots,M_j-1;\nonumber \\
	c&=& 0, \ldots,C-1.	 \nonumber
\end{eqnarray}
Then $(\bm{\eta}^i,\bm{\Theta}^i)$ are identifiable if and only if $(\bm{\epsilon}^i$, $\bm{\omega}^i)$ are identifiable.
\end{lemma}

The proof of Lemma \ref{tau and tau 0} is presented in Section B.

\begin{proof}[Proof of Proposition \ref{necessary c3}]
 	We first prove the second part of Proposition \ref{necessary c3} that Condition \ref{C3} is necessary for the identifiability of RegLCMs without covariates under Conditions \ref{A1} and \ref{C2}.  
 	 It is equivalent to show that if $\bm{\psi_0}, \cdots, \bm{\psi_{C-1}}$ are not linearly independent, $(\bm{\eta}, \bm{\Theta})$ are not identifiable. 
 	We prove it by the method of contradiction and assume the contrary that $\bm{\eta}$ are identifiable. Recall that the definitions in Section \ref{Models and Existing Works}, $ \bm{\eta} = (\eta_0, \cdots, \eta_{C-1})^{T}$ denotes the latent class membership probability, where $\eta_c = P(L = c)$ for $c = 0,\cdots, C-1$. And  $\bm{\Psi} = \left(\bm{\psi}_0,\cdots,\bm{\psi}_{C-1}\right)$ denotes the marginal probability matrix,
where each entry ${\psi}_{\bm{r}c}$ in $\bm{\psi}_c$ corresponding to a response pattern $\bm{r} \in \mathcal{S}^{\prime}$ is written as
\begin{equation}
{\psi}_{\bm{r}c} = P(\bm{R} = \bm{r} \mid L = c) = \prod_{j=1}^{J} \theta_{jr_jc}, \quad  c = 0,\cdots,C-1. \nonumber
\end{equation}
Based on the above definitions, we write the response probability vector as
\begin{equation}
 \left[P(\bm{R} = \bm{r}) : \bm{r} \in \mathcal{S}^{\prime}\right]^T = \Psi \bm{\cdot} \bm{\eta} . \label{response prob vector}
\end{equation}
As we assume $\bm{\eta}$ is identifiable, there exist no $\bm{\eta^{\prime}} \neq \bm{\eta}$ such that $P(\bm{R} = \bm{r} \mid \Psi, \bm{\eta}) = P(\bm{R} = \bm{r} \mid \Psi, \bm{\eta^{\prime}})$. According to \eqref{response prob vector}, $P(\bm{R} = \bm{r} \mid \Psi, \bm{\eta}) = P(\bm{R} = \bm{r} \mid \Psi, \bm{\eta^{\prime}})$ implies $\Psi \bm{\cdot} \bm{\eta} = \Psi \bm{\cdot} \bm{\eta^{\prime}}$.
However, under the condition that $\bm{\psi_0}, \cdots, \bm{\psi_{C-1}}$ are not linearly independent, there could exist $\bm{\eta^{\prime}}\neq \bm{\eta}$ such that $\Psi \bm{\cdot} (\bm{\eta} -\bm{\eta^{\prime}}) = \bm{0}$, and by the contradiction, $\bm{\eta}$ is not identifiable. 

 Next, we prove the first part of Proposition \ref{necessary c3}, the necessity of Condition \ref{A4} for the identifiability of RegLCMs under Conditions \ref{A1}--\ref{A3}. That is, if $\bm{\phi_0}, \cdots, \bm{\phi_{C-1}}$ are not linearly independent, then $(\bm{\beta}, \bm{\gamma}, \bm{\lambda})$ are not identifiable. This proof includes the following three steps.
 
 \textit{Step 1}: we prove if $\bm{\phi_0}, \cdots, \bm{\phi_{C-1}}$ are not linearly independent, then $\bm{\psi_0}^i, \cdots, \bm{\psi_{C-1}}^i $ are not linearly independent for $i = 1, \dots, N$, where each entry in $\bm{\psi_c}^i$ corresponds to a response pattern $\bm{r}  = (r_{1}, \cdots, r_{J}) \in \mathcal{S}^{\prime}$  and is defined as
 \begin{equation}
 	\bm{\psi}_{\bm{r}c}^i= P(\bm{R}_i = \bm{r} \mid L_i = c, \bm{x}_i, \bm{z}_i) =  \prod_{j=1}^{J} \frac {\text{exp}{(\gamma_{jr_jc} + \lambda_{1jr_j}z_{ij1} + \cdots + \lambda_{qjr_j}z_{ijq}})}{ 1 + \sum_{s = 1}^{M_j - 1} \text{exp}(\gamma_{jsc} + \lambda_{1js}z_{ij1} + \cdots + \lambda_{qjs} z_{ijq})}, \nonumber
 \end{equation} 
 for $c = 0, \dots, C-1$. Equivalently, we need to prove if there exists subject $i$ such that $\bm{\psi_0}^i, \cdots, \bm{\psi_{C-1}}^i $ are linearly independent, then $\bm{\phi_0}, \cdots, \bm{\phi_{C-1}}$ are linearly independent. We use similar techniques as in the \textit{Proof of Proposition 2} in \citeauthor{huang2004} \citeyear{huang2004}. First,                                                                                                                                                                                                                                                                                                                                                                                                                                                                                                             we associate the linear combinations of 
 $\bm{\phi}_c$'s  with $\bm{\psi}_c$'s as follows.  For any linear combination of $\bm{\phi}_c$'s with coefficients $a_c$'s, there exist $b_c$'s and $\bm{Y}^i$ such that the following equation holds,
   \begin{equation}
 	\sum_{c=0}^{C-1} a_c\bm{\phi}_c = \left(\sum_{c=0}^{C-1}b_c\bm{\psi}_c^i \right)\odot \bm{Y}^i, \label{li psi phi}
 \end{equation} 
 where $\odot$ denotes the element-wise multiplication and
 \begin{eqnarray}
 	 \bm{Y}^i  &=& \left(  \begin{matrix}
\prod_{j =1}^J {\frac{1}{\text{exp}{(\lambda_{1jr_j}z_{ij1} + \cdots + \lambda_{qjr_j}z_{ijq}})}} : \bm{r} =(r_1,\dots,r_J) \in \mathcal{S}^{\prime}
\end{matrix}\right)_{S \times 1}^T, \nonumber  \\
b_c &=&a_c \prod_{j=1}^{J} \frac{ 1 + \sum\limits_{s = 1}^{M_j - 1} \exp (\gamma_{jsc} + \lambda_{1js}z_{ij1} + \cdots + \lambda_{qjs}z_{ijq})}{1 + \sum\limits_{s = 1}^{M_j-1 } e^{\gamma_{jsc}}}. \label{bc ac}
 \end{eqnarray}
 To show  $\bm{\phi}_c$'s  are linearly independent, we need to show that  $\sum_{c=0}^{C-1} a_c\bm{\phi}_c = \bm{0}$ implies $a_0 = \cdots = a_{C-1} = 0$. Based on \eqref{li psi phi}, we have $\sum_{c=0}^{C-1} a_c\bm{\phi}_c = \bm{0}$ implies  $\sum_{c=0}^{C-1}b_c\bm{\psi}_c^i = \bm{0}$. Under the condition that $\bm{\psi}_0^i, \dots, \bm{\psi}_{C-1}^i$ are linearly independent, the equation
 \begin{equation}
 \label{psi li}
 \sum_{c=0}^{C-1} b_c\bm{\psi}_c^i = b_0\bm{\psi}_0^i + \cdots+b_{C-1}\bm{\psi}_{C-1}^i = \bm{0} 
 \end{equation}
 implies $b_0 = \cdots = b_{C-1} = 0$. And by \eqref{bc ac}, we have $a_0 = \cdots= a_{C-1} = 0$. Hence, $\bm{\phi}_0, \dots, \bm{\phi}_{C-1}$ are linearly independent when $\bm{\psi}_0^i, \dots, \bm{\psi}_{C-1}^i$ are linearly independent and we complete the proof for \textit{Step 1}.

\textit{Step 2}: We next introduce parameters $\epsilon_c^i$'s and $\omega_{jrc}^i$'s and show that they are not identifiable when $\bm{\psi_0}^i, \cdots, \bm{\psi_{C-1}}^i $ are not linearly independent. By the similar arguments in proving the necessity of Condition \ref{C3}, $(\bm{\eta}^i, \bm{\Theta}^i)$ are not identifiable when $\bm{\psi_0}^i, \cdots, \bm{\psi_{C-1}}^i $ are not linearly independent for any subject $i = 1,\dots, N$. Recall in RegLCMs, $(\bm{\eta}^i, \bm{\Theta}^i)$  are functionally dependent on the linear functions $\bm{x}_i^{T} \bm{\beta}$ and $\bm{\gamma}_{jc} + \bm{z}_{ij}^{T} \bm{\lambda}_j$, respectively, following the definitions of $(\bm{\eta}^i, \bm{\Theta}^i)$ and $(\bm{\beta}, \bm{\gamma}, \bm{\lambda})$ from \eqref{eta} and \eqref{theta} in main article. Next, we let
 \begin{equation}
 \epsilon_c^i = \bm{x_i}^{T}\bm{\beta}_c = \beta_{0c} + \beta_{1c}x_{i1} + \cdots + \beta_{pc}x_{ip}. \nonumber 
\end{equation}
for $i = 1, \ldots, N$, $c = 0, \dots, C-1$. And
\begin{equation}
 	\omega_{jrc}^i = {\gamma_{jrc}} + \bm{z}_{ij}^{T} \bm{\lambda}_{jr} = \gamma_{jrc} + \lambda_{1jr}z_{ij1} + \cdots + \lambda_{qjr}z_{ijq}. \nonumber 
 	\end{equation}
 	$\text{for}\ i = 1, \ldots, N$, $j =1, \dots, J$, $r = 0, \cdots, {M_j-1}$ and $c = 0,\dots, C-1$. Then according to Lemma \ref{tau and tau 0}, $\epsilon_c^i$'s and $\omega_{jrc}^i$'s are not identifiable when $(\bm{\eta}^i, \bm{\Theta}^i)$ are not identifiable. Hence, $\epsilon_c^i$'s and $\omega_{jrc}^i$'s are not identifiable when $\bm{\psi_0}^i, \cdots, \bm{\psi_{C-1}}^i $ are not linearly independent and we complete the proof for \textit{Step 2}.
 	 	
 	\textit{Step 3}: Lastly, we prove that $(\bm{\beta}, \bm{\gamma}, \bm{\lambda})$  are not identifiable when $\epsilon_c^i$'s and $\omega_{jrc}^i$'s are not identifiable by the method of contradiction. Assume to the contrary that $\bm{\beta}$ is identifiable given $\epsilon_c^i$'s and $\omega_{jrc}^i$'s are not identifiable. By the definition of identifiability,  $P(\bm{R} \mid \bm{\beta}^{*}, \bm{\gamma}, \bm{\lambda} ) = P(\bm{R} \mid \bm{\beta}^{\prime},  \bm{\gamma}, \bm{\lambda}) $ implies that $\bm{\beta}^{*} = \bm{\beta}^{\prime}$. Because $\bm{X}$ has full column rank and according to the definition of $\bm{\epsilon}$ that
  \begin{equation}
\bm{\epsilon} = \left(  \begin{array}{c}
 \bm{\epsilon}^{1} \\
\vdots \\
\bm{\epsilon}^{N}
\end{array}\right)
= \left(  \begin{array}{cccc}
1 & x_{11} &\cdots& x_{1p}\\
\vdots & \vdots & \ddots & \vdots\\
1 & x_{N1} &\cdots & x_{Np}
\end{array}\right) \left( \begin{array}{ccc}
\beta_{00} & \cdots & \beta_{0(C-1)}\\
\vdots & \ddots &  \vdots \\
\beta_{p0} & \cdots  & \beta_{p(C-1)}
\end{array}\right) = \bm{X}\bm{\beta}, \nonumber
\end{equation}
we have $\bm{\epsilon}^* = \bm{X}\bm{\beta}^*$ equivalent to $\bm{\epsilon}^{\prime} = \bm{X}\bm{\beta}^{\prime}$. So for all subject $i$, $P(\bm{R}_i \mid \bm{\epsilon}^{i*}, \bm{\gamma}, \bm{\lambda} ) = P(\bm{R}_i \mid \bm{\epsilon}^{i\prime},  \bm{\gamma}, \bm{\lambda})$ would result in $\bm{\epsilon}^{i*} = \bm{\epsilon}^{i\prime}$, which contradicts the non-identifiability of ${\epsilon}_c^i$'s. Therefore $\bm{\beta}$ is not identifiable. Using similar techniques, we can prove $\bm{\gamma}, \bm{\lambda}$ are not identifiable. 

Combining \textit{Steps 1--3}, we prove the first part of Proposition \ref{necessary c3}, and thus complete the proof of Proposition \ref{necessary c3}.
 \end{proof}

\begin{proof}[Proof of Proposition \ref{non identifiability but full column rank}]
First, we show that for polytomous-response CDMs or RegCDMs, the parameters are not generically identifiable under Condition \ref{P1} that some attribute is required by only one item. This is motivated from the proof of Theorem 4.4 (a) in \citeauthor{gu2018partial} \citeyear{gu2018partial}, where they showed that the binary-response CDMs are not generically identifiable if some attribute is required by only one item. Consider the polytomous-response CDMs and let the $Q$-matrix to be
\[
 			Q = \left(\begin{array}{cccc}
					1& \bm{u} \\  
					\bm{0} & Q^{*} \\
 					\end{array} \right).
 			\]	
This $Q$-matrix implies that $\alpha_1$ is required by the first item only. For any $(\bm{\eta}, \bm{\Theta})$, we can construct  $(\bar{\bm{\eta}}, \bar{\bm{\Theta}}) \neq (\bm{\eta}, \bm{\Theta}) $ such that $P(\bm{R} = \bm{r} \mid  \bm{\eta} , \bm{\Theta}) = P(\bm{R} = \bm{r} \mid \bar{\bm{\eta}}, \bar{\bm{\Theta}})$, which shows that $(\bm{\eta}, \bm{\Theta})$ are not identifiable. To better illustrate the idea, we next use $\bm{\alpha}$ to replace $c$ in all parameter subscripts, i.e. $\eta_{\bm{\alpha}} = \eta_c $ and $\theta_{jr\bm{\alpha}} = \theta_{jrc}$ given $\bm{\alpha}^{T} \bm{v} = c$. When $j \neq 1$, we let $\eta_c = \bar{\eta}_c$, $\theta_{jrc} = \bar{\theta}_{jrc}$ for $r = 0,\ldots, M_j -1$ and $c = 0, \ldots, C-1$.  When $j = 1$, we denote $\bm{\alpha}^{\prime} = ({\alpha}_{2}, \cdots, \alpha_K) \in \{0,1\}^{K-1}$ and for all $r_1 = 0, \dots, M_1 -1$, we let $\bar{\theta}_{1r_1(0,\bm{\alpha}^{\prime})} = {\theta}_{1r_1(0,\bm{\alpha}^{\prime})}$, and 
\begin{equation}
	\bar{\theta}_{1r_1(1,\bm{\alpha}^{\prime})} = \frac{1}{E} {\theta}_{1r_1(1,\bm{\alpha}^{\prime})} + (1-\frac{1}{E}) {\theta}_{1r_1(0,\bm{\alpha}^{\prime})}, \nonumber
\end{equation}
where $E$ is a constant in a small neighborhood of $1$ and $E \neq 1$. So we have $\bar{\theta}_{1r_1(1,\bm{\alpha}^{\prime})} \neq {\theta}_{1r_1(1,\bm{\alpha}^{\prime})}$. We also let
\begin{eqnarray}
 \bar{\eta}_{(0,\bm{\alpha}^{\prime})} &=& {\eta}_{(0,\bm{\alpha}^{\prime})} + (1-E)\cdot {\eta}_{(1,\bm{\alpha}^{\prime})},\nonumber \\
\bar{\eta}_{(1,\bm{\alpha}^{\prime})} &=& E\cdot {\eta}_{(1,\bm{\alpha}^{\prime})}. \nonumber
\end{eqnarray}
Hence, we have
\begin{eqnarray}
	\bar{\eta}_{(1,\bm{\alpha}^{\prime})} + \bar{\eta}_{(0,\bm{\alpha}^{\prime})} &=& {\eta}_{(1,\bm{\alpha}^{\prime})} + {\eta}_{(0,\bm{\alpha}^{\prime})}, \\
	\bar{\theta}_{1r_1(1,\bm{\alpha}^{\prime})}\bar{\eta}_{(1,\bm{\alpha}^{\prime})} + \bar{\theta}_{1r_1(0,\bm{\alpha}^{\prime})}\bar{\eta}_{(0,\bm{\alpha}^{\prime})} &=& {\theta}_{1r_1(1,\bm{\alpha}^{\prime})}{\eta}_{(1,\bm{\alpha}^{\prime})} + {\theta}_{1r_1(0,\bm{\alpha}^{\prime})}{\eta}_{(0,\bm{\alpha}^{\prime})}. \label{theta eta bar equal theta eta}
\end{eqnarray}
So for any $\bm{r} = (r_1,\cdots, r_J) \in \mathcal{S}^{\prime} $, we use $\bm{\psi}_{r,\cdot}$ to denote the row in $\bm{\Psi}$ corresponding to the response pattern $\bm{r}$. By the definition of the conditional response probability,  we write 
\begin{eqnarray}
	&& P(\bm{R} = \bm{r} \mid \bar{\Psi}, \bar{\bm{\eta}} ) = \bar{\bm{\psi}}_{r,\cdot} \cdot \bar{\bm{\eta}}  = \sum\limits_{\bm{\alpha} \in \{0,1\}^K } {\psi}_{{r}\bm{\alpha}} {\eta}_{\bm{\alpha}}   \nonumber \\
	&=&  \sum_{\substack{ \bm{\alpha}^{\prime} \in \{0,1\}^{K-1} \\ \alpha_1 \in \{0,1\}}} \prod_{j>1} \{\bar{\theta}_{jr_j(\alpha_1,\bm{\alpha}^{\prime})}\}^{\mathbb{I}\{R_j = r_j\}}\bar{\eta}_{(\alpha_1,\bm{\alpha}^{\prime})} [\{\bar{\theta}_{1r_1(1,\bm{\alpha}^{\prime})}\}^{\mathbb{I}\{R_1 = r_1\}}\bar{\eta}_{(1,\bm{\alpha}^{\prime})} + \{\bar{\theta}_{1r_1(0,\bm{\alpha}^{\prime})}\}^{\mathbb{I}\{R_1 = r_1\}}\bar{\eta}_{(0,\bm{\alpha}^{\prime})}] \nonumber \\
	&=& 
	\begin{cases}
 	\sum\limits_{\substack{ \bm{\alpha}^{\prime} \in \{0,1\}^{K-1} \\ \alpha_1 \in \{0,1\}}} \prod\limits_{j>1} \{\bar{\theta}_{jr_j(\alpha_1,\bm{\alpha}^{\prime})}\}^{\mathbb{I}\{R_j = r_j\}}\bar{\eta}_{(\alpha_1,\bm{\alpha}^{\prime})} [\bar{\theta}_{1r_1(1,\bm{\alpha}^{\prime})}\bar{\eta}_{(1,\bm{\alpha}^{\prime})} + \bar{\theta}_{1r_1(0,\bm{\alpha}^{\prime})}\bar{\eta}_{(0,\bm{\alpha}^{\prime})}], & \text{$R_1 = r_1$} \\
 	\sum\limits_{\substack{ \bm{\alpha}^{\prime} \in \{0,1\}^{K-1} \\ \alpha_1 \in \{0,1\}}} \prod\limits_{j>1} \{\bar{\theta}_{jr_j(\alpha_1,\bm{\alpha}^{\prime})}\}^{\mathbb{I}\{R_j = r_j\}}\bar{\eta}_{(\alpha_1,\bm{\alpha}^{\prime})} [\bar{\eta}_{(1,\bm{\alpha}^{\prime})} + \bar{\eta}_{(0,\bm{\alpha}^{\prime})}], & \text{$R_1 \neq r_1$}
 	\end{cases} \nonumber \\
	&=& 
	\begin{cases}
	\label{eq22}
 	\sum\limits_{\substack{ \bm{\alpha}^{\prime} \in \{0,1\}^{K-1} \\ \alpha_1 \in \{0,1\}}} \prod\limits_{j>1} \{{\theta}_{jr_j(\alpha_1,\bm{\alpha}^{\prime})}\}^{\mathbb{I}\{R_j = r_j\}}{\eta}_{(\alpha_1,\bm{\alpha}^{\prime})} [{\theta}_{1r_1(1,\bm{\alpha}^{\prime})}{\eta}_{(1,\bm{\alpha}^{\prime})} + {\theta}_{1r_1(0,\bm{\alpha}^{\prime})}{\eta}_{(0,\bm{\alpha}^{\prime})}], & \text{$R_1 = r_1$} \\
 	\sum\limits_{\substack{ \bm{\alpha}^{\prime} \in \{0,1\}^{K-1} \\ \alpha_1 \in \{0,1\}}} \prod\limits_{j>1} \{{\theta}_{jr_j(\alpha_1,\bm{\alpha}^{\prime})}\}^{\mathbb{I}\{R_j = r_j\}}{\eta}_{(\alpha_1,\bm{\alpha}^{\prime})} [{\eta}_{(1,\bm{\alpha}^{\prime})} + {\eta}_{(0,\bm{\alpha}^{\prime})}], & \text{$R_1 \neq r_1$}
 	\end{cases}  \\
	&=&  \sum_{\substack{ \bm{\alpha}^{\prime} \in \{0,1\}^{K-1} \\ \alpha_1 \in \{0,1\}}} \prod_{j>1} \{{\theta}_{jr_j(\alpha_1,\bm{\alpha}^{\prime})}\}^{\mathbb{I}\{R_j = r_j\}}{\eta}_{(\alpha_1,\bm{\alpha}^{\prime})} [\{{\theta}_{1r_1(1,\bm{\alpha}^{\prime})}\}^{\mathbb{I}\{R_1 = r_1\}}{\eta}_{(1,\bm{\alpha}^{\prime})} + \{{\theta}_{1r_1(0,\bm{\alpha}^{\prime})}\}^{\mathbb{I}\{R_1 = r_1\}}{\eta}_{(0,\bm{\alpha}^{\prime})}] \nonumber \\
	&=& P(\bm{R} = \bm{r} \mid {\Psi}, {\bm{\eta}} ). \nonumber
\end{eqnarray}
Equation (\ref{eq22}) is derived based on (\ref{theta eta bar equal theta eta}) as well as the assumption that $\eta_c = \bar{\eta}_c$, $\theta_{jrc} = \bar{\theta}_{jrc}$ for all $j = 2, \dots, J$, $r = 0,\ldots, M_j -1$ and $c = 0, \ldots, C-1$. With this construction, we show different $(\bm{\eta}, \bm{\Theta}) $ could result in the same conditional response probability and therefore we prove that  $(\bm{\eta}, \bm{\Theta})$ are not identifiable under Condition \ref{P1} in Proposition \ref{non identifiability but full column rank}. 

For polytomous-response RegCDMs, we have similar results by following the above proof. That is, $(\bm{\eta}^i, \bm{\Theta}^i)$ are not identifiable under Condition \ref{P1} for  $i = 1,\dots, N$. Then following the same arguments as in \textit{Steps 2--3} from the \textit{Proof of Proposition \ref{necessary c3}}, we show that $(\bm{\beta}, \bm{\gamma}, \bm{\lambda})$ in RegCDMs are not identifiable given $(\bm{\eta}^i, \bm{\Theta}^i)$ are not identifiable.

Next we prove the remaining part, that is, the matrix $\bm{\Psi}$ in CDMs and the matrix $\bm{\Phi}$ in RegCDMs have full column ranks under Condition \ref{P2}. 
Before presenting the  proof, we introduce another probability matrix $T$-matrix of size $S \times C$, where each row corresponds to one response pattern $\bm{r}\in \mathcal{S}$ and each column  corresponds to one latent class $c = 0, \dots, C-1$. Each entry of $T$-matrix is defined as $
T_{\bm{r}c} = P (\bm{R} \succeq \bm{r} \mid L = c),
$
where $\succeq$ means that for any item $j = 1,\dots, J$,  $R_j \geq r_j$. According to a similar argument in Appendix Section 4.2 in \citeA{xu2017}, $T$-matrix has full column rank under the condition that the corresponding $Q$-matrix contains an identity submatrix $\mathcal{I}_K$.

There exists a relation between the two probability matrices, $T$-matrix and $\bm{\Psi}$. Because $\bm{\Psi}$ excludes a reference response pattern, it has dimension of $(S -1) \times C $. Denote $\bm{\Psi}^{\prime} = (\bm{\Psi}^T, \bm{\Psi}_{ref}^T)^T$ where $\bm{\Psi}_{ref}$  is the row corresponding to the reference pattern. And $\bm{\Psi}_{ref}$ is linearly dependent on the rows in $\bm{\Psi}$ because $\sum_{\bm{r} \in \mathcal{S}}P(\bm{R} = \bm{r} \mid L = c) = 1$. So $\bm{\Psi}$ has full column rank if and only if $\bm{\Psi}^{\prime}$ has full column rank. Further, $\bm{\Psi}^{\prime}$ has full column rank if and only if $T$-matrix has full column rank, because $\bm{\Psi}^{\prime}$ is bijectively corresponding to $T$-matrix according to their definitions. In conclusion, $\bm{\Psi}$ in the CDMs has full column rank when $Q$-matrix contains an identity submatrix $\mathcal{I}_K$. According to the \textit{Proof of Proposition 2} in \citeA{huang2004}, the matrix $\bm{\Phi}$ has full column rank when the matrix $\bm{\Psi}$ has full column rank. So for RegCDMs, $\bm{\Phi}$ has full column rank when $Q$-matrix contains an identity submatrix $\mathcal{I}_K$.
\end{proof}
\begin{proof}[Proof of Theorem \ref{local RegLCMs}]

 Following the similar idea in \citeA{huang2004} at page 15, we let $f(\bm{R} ; \bm{\eta},\bm{\Theta})$ to denote the likelihood function, and
\begin{eqnarray}
	f(\bm{R} ; \bm{\eta},\bm{\Theta}) = \prod_{\bm{r}\in \mathcal{S} }P(\bm{R} = \bm{r})^{\mathbb{I}\{\bm{R} = \bm{r}\}}, \nonumber
\end{eqnarray}
where $\mathcal{S} = \bigtimes_{j=1}^{J} \{0,\dots,M_j-1\}$. Let $ \bm{\xi} =  \{ {\eta_1}, \cdots,{\eta_{C-1}}, {\theta_{110}}, \cdots, {\theta_{1(M_1-1)0}},\cdots,{\theta_{J1(C-1)}}, \cdots,$ ${\theta_{J(M_J-1)(C-1)}} \}$,  the Fisher information matrix is written as
\begin{eqnarray}
	&&\mathbb{E}\left[\left(\frac{\partial \log f}{\partial \bm{\xi}}\right)\left(\frac{\partial \log f}{\partial \bm{\xi}}\right)^{T}\right]\nonumber \\
	 &=& \mathbb{E}\left[\left(\sum_{\bm{r}\in \mathcal{S}}\frac{\mathbb{I}\{\bm{R}= \bm{r}\}}{P(\bm{R}= \bm{r})}\frac{\partial  P(\bm{R}= \bm{r})}{\partial \bm{\xi}}\right)\left(\sum_{\bm{r}\in \mathcal{S}}\frac{\mathbb{I}\{\bm{R}= \bm{r}\}}{P(\bm{R}= \bm{r})}\frac{\partial  P(\bm{R}= \bm{r})}{\partial \bm{\xi}}\right)^{T}\right]  \nonumber \\
	 &=& \sum_{\bm{r}\in\mathcal{S}}\frac{1}{P(\bm{R}= \bm{r})}\left(\frac{\partial  P(\bm{R}= \bm{r})}{\partial \bm{\xi}}\right)\left(\frac{\partial  P(\bm{R}= \bm{r})}{\partial \bm{\xi}}\right)^{T} \nonumber \\
	 &=& \mathbf{J}^T  \left(\begin{array}{cccc}
					\frac{1}{P(\bm{R}= \bm{r}_1)}& 0 & \cdots & 0 \\  
					0& \frac{1}{P(\bm{R}= \bm{r}_2)} &  \cdots &0 \\
					\vdots & \vdots & \ddots & \vdots \\
					0 & 0 & \cdots & \frac{1}{P(\bm{R}= \bm{r}_S)} \\
 					\end{array} \right) \mathbf{J}. \nonumber
\end{eqnarray}
From the above results, we see the Fisher information matrix is non-singular if and only if $\mathbf{J}$ has full column rank. According to Theorem 1 of \citeA{Rothenberg1971},  $(\bm{\eta},\bm{\Theta})$ are locally identifiable if and only if the Fisher information matrix is non-singular when the true values of $(\bm{\eta},\bm{\Theta})$ are regular points of the information matrix. Therefore $(\bm{\eta},\bm{\Theta})$ are locally identifiable if and only if the Jacobian matrix $\mathbf{J}$ has full column rank.
\end{proof}
\begin{proof}[Proof of Theorem \ref{local RegCDMs}]
As introduced in Section \ref{proposed identifiability conditions}, we consider a hypothetical subject with all covariates being zeros and denote its Jacobian matrix as $\mathbf{J}^0$. We use the following three steps to prove that $(\bm{\beta}, \bm{\gamma}, \bm{\lambda})$ are identifiable if and only if $\mathbf{J}^0$ has full column rank.

\textit{Step 1}: We first show that for  subject $i= 1, \dots, N$, the Jacobian matrices $\mathbf{J}^i$, containing  the derivatives of conditional response probabilities with respect to parameters in $\bm{\eta}^i$ and $\bm{\Theta}^i$, have full column rank if and only if $\mathbf{J}^0$ has full column rank. This proof is adapted from the \textit{Proof of Proposition 1} in \citeA{huang2004}. 
	 
	 First, we need to set up a few notations. The Jacobian matrix $\mathbf{J}^i$ is written as
	\begin{equation}
\mathbf{J}^i=\left(\bm{J}_{\eta_1}^i, \cdots,\bm{J}_{\eta_{C-1}}^i, \bm{J}_{\theta_{110}}^i, \cdots, \bm{J}_{\theta_{1(M_1-1)0}}^i,\cdots,\bm{J}_{\theta_{J1(C-1)}}^i, \cdots, \bm{J}_{\theta_{J(M_J-1)(C-1)}}^i \right). \nonumber
\end{equation}
Each entry in $\bm{J}_{\eta_c}^i$ is a partial derivative of response probability $P(\bm{R} = \bm{r})$  with respect to $\eta_c^i$ at true value of $\eta_c^i$, which is computed to be
\begin{eqnarray}
	\frac{\partial{P(\bm{R} = \bm{r})}}{\partial{\eta_c^i}} &=& \prod_{j=1}^{J}\theta_{jr_jc}^{i} - \prod_{j=1}^{J}\theta_{jr_j0}^i = \bm{\psi}_{\bm{r}c}^i - \bm{\psi}_{\bm{r}0}^i. \nonumber
\end{eqnarray}
And each entry in $\bm{J}_{\theta_{jrc}}^i$ is a partial derivative of response probability $P(\bm{R} = \bm{r})$ with respect to $\theta_{jrc}^i$ at true value of $\theta_{jrc}^i$, which is computed to be
	\begin{eqnarray}
	\frac{\partial{P(\bm{R} = \bm{r})}}{\partial{\theta_{jrc}^i}} &=&
  \begin{cases}
     \eta_c^i \prod_{d \neq j}\theta_{dr_dc}^i, & \text{if $r_{j} = r$}, \\
     - \eta_c^i \prod_{d \neq j} \theta_{dr_dc}^i, & \text{if $r_{j} = 0$}, \\
  0, & \text{otherwise}.
  \end{cases}   \nonumber
\end{eqnarray}
or summarized as 
\begin{equation}
	\frac{\partial{P(\bm{R} = \bm{r})}}{\partial{\theta_{jrc}^i}} = \eta_c^i \bm{\psi}_{\bm{r}c}^i(\frac{\mathbb{I}\{r_{j} = r\}}{\theta_{jrc}^i} - \frac{\mathbb{I}\{r_{j} = 0\}}{\theta_{j0c}^i} ). \nonumber
\end{equation}
	In addition to $\mathbf{J}^i$, we also define the following two sets of vectors for this proof. Denote $\overline{\mathbf{J}}^{0} =\{\bm{\psi}_0^0, \dots, \bm{\psi}_{C-1}^0\} \cup\{ \eta_{c}^0\left({\mathbf{I}\{{r}_{j} = r\}}/{\theta_{jrc}^0}  \right) \odot \bm{\psi}_{c}^0: j = 1,\dots,J, r=0,\dots,M_j-1,c=0,\dots,C-1 \}$ and $\overline{\mathbf{J}}^i = \{\bm{\psi}_0^i, \dots, \bm{\psi}_{C-1}^i\} \cup\{ \eta_{c}^i\left({\mathbf{I}\{{r}_{j} = r\}}/{\theta_{jrc}^i}  \right) \odot \bm{\psi}_{c}^i: j = 1,\dots,J, r=0,\dots,M_j-1,c=0,\dots,C-1 \}$, where $\mathbf{I}\{{r}_{j} = r\}$ is a ($S-1$)-dimensional vector containing all $\mathbb{I}\{r_j = r\}$ for $\bm{r}=(r_1,\dots,r_J) \in \mathcal{S}^{\prime}$. With the notations defined, we then introduce a useful lemma which simplify the arguments in proving the linear independence of the columns in $\mathbf{J}^0$ and $\mathbf{J}^i$.
	
	\begin{lemma}
	\label{lemma_in_thm2_step1}  
	The Jacobian matrix 
		$\mathbf{J}^0$ has full column rank if and only if $\overline{\mathbf{J}}^{0}$
		 are linearly independent. The Jacobian matrix  $\mathbf{J}^i$ has full column rank if and only if $\overline{\mathbf{J}}^i$ are linearly independent.
	\end{lemma}
	
	The proof of Lemma \ref{lemma_in_thm2_step1} is presented in Section B. According to Lemma \ref{lemma_in_thm2_step1}, to prove $\mathbf{J}^i$ has full column rank if and only if $ \mathbf{J}^0$ has full column rank, we can equivalently show that $\overline{\mathbf{J}}^i$ are linearly independent if and only if $\overline{\mathbf{J}}^{0}$ are linearly independent. First, we associate the linear combinations of $\overline{\mathbf{J}}^i$ to that of $\overline{\mathbf{J}}^{0}$ as follows. For any linear combinations of   $\overline{\mathbf{J}}^{i} $ with coefficients $t_c^i$, $u_{jrc}^i$, there exist $t_c^0$, $u_{jrc}^0$ and $\bm{W}^i$ such that the following equation holds 
\begin{eqnarray}
	&& \sum_{c=0}^{C-1}t_{c}^i\bm{\psi}_{c}^i + \sum_{j=1}^{J} \sum_{r=0}^{M_j-1} \sum_{c=0}^{C-1}\left\{u_{jrc}^i \eta_{c}^i\left(\frac{\mathbf{I}\{{r}_{j} = r\}}{\theta_{jrc}^i}  \right) \odot \bm{\psi}_{c}^i\right\}  \nonumber  \\
&&	 = \left(\sum_{c=0}^{C-1}t_{c}^0\bm{\psi}_{c}^0   + \sum_{j=1}^{J} \sum_{r=0}^{M_j-1} \sum_{c=0}^{C-1}\left\{ u_{jrc}^0 \eta_{c}^0\left(\frac{\mathbf{I}\{{r}_{j} = r\}}{\theta_{jrc}^0} \right) \odot \bm{\psi}_{c}^0\right\}\right) \odot  \bm{W}^i,  \label{eq ti t0}
\end{eqnarray}
where \begin{eqnarray}
		 \bm{W}^i  &=& \left(  \begin{matrix}
\prod_{j =1}^J {\text{exp}{(\lambda_{1jr_j}z_{ij1} + \cdots + \lambda_{qjr_j}z_{ijq}})} : \bm{r} =(r_1,\dots,r_J) \in \mathcal{S}^{\prime}
\end{matrix}\right)_{S \times 1}^T, \nonumber  \\
	t_c^0 &=& t_c^i \prod_{j=1}^{J} \frac{1 + \sum\limits_{s = 1}^{M_j-1 } e^{\gamma_{jsc}}}{ 1 + \sum\limits_{s = 1}^{M_j - 1} \exp (\gamma_{jsc} + \lambda_{1js}z_{ij1} + \cdots + \lambda_{qjs}z_{ijq})} ,\label{tc0} \\
	u_{jrc}^0 &=& u_{jrc}^i  \frac{ \text{exp}( \beta_{1c}x_{i1} + \cdots + \beta_{pc}x_{ip})} {\text{exp}{(\lambda_{1jr_j}z_{ij1} + \cdots + \lambda_{qjr_j}z_{ijq}})} \nonumber \\
	&& \times \frac{\{1+ \sum\limits_{l=1}^{C-1}e^{\beta_{0l}}\}\{1 + \sum\limits_{s = 1}^{M_j - 1} \exp (\gamma_{jsc} + \lambda_{1js}z_{ij1} + \cdots + \lambda_{qjs}z_{ijq})\} }{\{1 + \sum\limits_{l=1}^{C-1}\text{exp}(\beta_{0l} + \beta_{1l}x_{i1} + \cdots + \beta_{pl}x_{ip})\} \{1 + \sum\limits_{s = 1}^{M_j-1 } e^{\gamma_{jsc}}\} } \nonumber \\
	&&  \times \prod_{j = 1}^J  \frac{1 + \sum\limits_{s = 1}^{M_j-1 } e^{\gamma_{jsc}} }{ 1 + \sum\limits_{s = 1}^{M_j - 1} \exp (\gamma_{jsc} + \lambda_{1js}z_{ij1} + \cdots + \lambda_{qjs}z_{ijq}) } .\label{ujrc0}
	\end{eqnarray}
 
 The next two parts prove that $\overline{\mathbf{J}}^i$ are linearly independent if and only if $\overline{\mathbf{J}}^{0}$ are linearly independent in two directions.

\textit{Part (i)}: We prove $\overline{\mathbf{J}}^i$ are linearly independent if $\overline{\mathbf{J}}^{0}$ are linearly independent. To show $\overline{\mathbf{J}}^i$ are linearly independent, we need to show that \begin{equation}
	\sum_{c=0}^{C-1}t_{c}^i\bm{\psi}_{c}^i + \sum_{j=1}^{J} \sum_{r=0}^{M_j-1} \sum_{c=0}^{C-1}\left\{u_{jrc}^i \eta_{c}^i\left(\frac{\mathbf{I}\{{r}_{j} = r\}}{\theta_{jrc}^i}  \right) \odot \bm{\psi}_{c}^i\right\} = \bm{0}, \label{Ji bar li}
	\end{equation}
implies $t_c^i = 0$ and $u_{jrc}^i=0$. By \eqref{eq ti t0}, for any $t_c^i$, $u_{jrc}^i$ such that \eqref{Ji bar li} holds,	we have 
	\begin{equation}
		\sum_{c=0}^{C-1}t_{c}^0\bm{\psi}_{c}^0+ \sum_{j=1}^{J} \sum_{r=0}^{M_j-1} \sum_{c=0}^{C-1}\left\{u_{jrc}^0 \eta_{c}^0(\frac{\mathbb{I}\{\bm{r}_{j} = r\}}{\theta_{jrc}^0}  ) \odot \bm{\psi}_{c}^0\right\}=\bm{0}. \nonumber 
	\end{equation}
Under the condition that $\overline{\mathbf{J}}^{0}$ are linearly independent, we have $ t_{c}^0 =0$ and $ u_{jrc}^0=0$. Then by \eqref{tc0} and \eqref{ujrc0}, we have $t_{c}^i = 0$ and $u_{jrc}^i=0$ for $j=1,\ldots,J$, $r = 0, \ldots, M_j-1$ and $c = 0,\ldots, C-1$. So $\overline{\mathbf{J}}^i$ are linearly independent.

 \textit{Part (ii)}: We prove $\overline{\mathbf{J}}^0$ are linearly independent if $\overline{\mathbf{J}}^{i}$ are linearly independent. This part is similar to  \textit{Part (i)}. To show $\overline{\mathbf{J}}^0$ are linearly independent, we need to show that 
\begin{equation}
	\sum_{c=0}^{C-1}t_{c}^0\bm{\psi}_{c}^0 + \sum_{j=1}^{J} \sum_{r=0}^{M_j-1} \sum_{c=0}^{C-1}\left\{u_{jrc}^0 \eta_{c}^0\left(\frac{\mathbf{I}\{{r}_{j} = r\}}{\theta_{jrc}^0}  \right) \odot \bm{\psi}_{c}^0\right\} = \bm{0} \label{J0 bar li}
\end{equation}
implies $t_c^0 = 0$ and $u_{jrc}^0 = 0$. By \eqref{eq ti t0}, for any $t_c^0, u_{jrc}^0$ such that \eqref{J0 bar li} holds, we have
\begin{equation}
	\sum_{c=0}^{C-1}t_{c}^i\bm{\psi}_{c}^i+ \sum_{j=1}^{J} \sum_{r=0}^{M_j-1} \sum_{c=0}^{C-1}\left\{u_{jrc}^i \eta_{c}^i(\frac{\mathbf{I}\{\bm{r}_{j} = r\}}{\theta_{jrc}^i}  ) \odot \bm{\psi}_{c}^i\right\}=\bm{0}. \label{ai bi eq 0}  \nonumber
\end{equation}
Under the condition that $\overline{\mathbf{J}}^{i}$ are linearly independent,  $ t_{c}^i=0$ and $ u_{jrc}^i=0$, and hence $t_{c}^0 = 0$ and  $ u_{jrc}^i=0$ by \eqref{tc0} and \eqref{ujrc0}, for $j=1,\ldots,J$, $r = 0, \ldots, M_j-1$ and $c = 0,\ldots, C-1$.  So $\overline{\mathbf{J}}^{0}$ are linearly independent.

 Combining \textit{Part (i)} and \textit{Part (ii)}, we show $\overline{\mathbf{J}}^i$ are linearly independent if and only if $\overline{\mathbf{J}}^{0}$ are linearly independent. And therefore $\mathbf{J}^i$ has full column rank if and only if $\mathbf{J}^0$ has full column rank.

\textit{Step 2}: We introduce $(\bm{\epsilon}^i$, $\bm{\omega}^i)$ and prove that they are identifiable if and only if $\mathbf{J}^i$ has full column rank. By following  similar arguments in the \textit{Proof of Theorem \ref{local RegLCMs}}, we have $(\bm{\eta}^i, \bm{\Theta}^i)$ are identifiable if and only if $\mathbf{J}^i$ has full column rank, for $i = 1, \dots, N$. Next, we define $(\bm{\epsilon}^i$, $\bm{\omega}^i)$ and the remaining is to show that they are identifiable if and only if $(\bm{\eta}^i, \bm{\Theta}^i)$ are identifiable. Following the same arguments as \textit{Step 2} in \textit{Proof of Proposition \ref{necessary c3}}, we let
	 \begin{equation}
 \epsilon_c^i = \bm{x_i}^{T}\bm{\beta}_c = \beta_{0c} + \beta_{1c}x_{i1} + \cdots + \beta_{pc}x_{ip}. \nonumber 
\end{equation}
for $c = 0, \dots, C-1$. And
\begin{equation}
 	\omega_{jrc}^i = {\gamma_{jrc}} + \bm{z}_{ij}^{T} \bm{\lambda}_{jr} = \gamma_{jrc} + \lambda_{1jr}z_{ij1} + \cdots + \lambda_{qjr}z_{ijq}. \nonumber 
 	\end{equation}
 	$\text{for}\ j =1, \dots, J$, $r = 0, \cdots, {M_j-1}$ and $c = 0,\dots, C-1$. 
	Then according to Lemma \ref{tau and tau 0}, $(\bm{\epsilon}^i$, $\bm{\omega}^i)$ are identifiable if and only if $(\bm{\eta}^i, \bm{\Theta}^i)$ are identifiable. Hence the proof for \textit{Step 2} is complete.
	
	\textit{Step 3}:
	The final step is to show $(\bm{\beta}, \bm{\gamma}, \bm{\lambda})$ are identifiable if and only if $(\bm{\epsilon}^i$, $\bm{\omega}^i)$ are identifiable. We have shown that $(\bm{\beta}, \bm{\gamma}, \bm{\lambda})$ are not identifiable when $(\bm{\epsilon}^i$, $\bm{\omega}^i)$ are not identifiable in the \textit{Proof of Proposition \ref{necessary c3}}. So all left to show is the necessary part that $(\bm{\beta}, \bm{\gamma}, \bm{\lambda})$  are identifiable when $(\bm{\epsilon}^i$, $\bm{\omega}^i)$ are identifiable. We prove this result by the method of contradiction. Assuming the contrary that $\bm{\beta}$ is not identifiable,  there exist $\bm{\beta} \neq \bm{\beta}^{\prime}$ such that $P(\bm{R}_i \mid \bm{\beta}, \bm{\gamma}, \bm{\lambda} ) = P(\bm{R}_i \mid \bm{\beta}^{\prime}, \bm{\gamma}, \bm{\lambda} ) $.
	According to the system of linear equations
  \begin{equation}
\bm{\epsilon} = \left(  \begin{array}{c}
 \bm{\epsilon}^{1} \\
\vdots \\
\bm{\epsilon}^{N}
\end{array}\right)
= \left(  \begin{array}{cccc}
1 & x_{11} &\cdots& x_{1p}\\
\vdots & \vdots & \ddots & \vdots\\
1 & x_{N1} &\cdots & x_{Np}
\end{array}\right) \left( \begin{array}{ccc}
\beta_{00} & \cdots & \beta_{0(C-1)}\\
\vdots & \ddots &  \vdots \\
\beta_{p0} & \cdots  & \beta_{p(C-1)}
\end{array}\right) = \bm{X}\bm{\beta}, \nonumber
\end{equation}
and because the full rank $\bm{X}$ is an injective mapping, we have $\bm{\beta} \neq \bm{\beta}^{\prime}$ implies that $\bm{\epsilon} = \bm{X}\bm{\beta}$ is different from $\bm{\epsilon}^{\prime} = \bm{X}\bm{\beta}^{\prime}$ for at least one $\bm{\epsilon}^i \neq \bm{\epsilon}^{\prime i}$. However,	since $\bm{\epsilon}^i$'s are identifiable, there exist no $\bm{\epsilon}^i \neq \bm{\epsilon}^{\prime i}$ such that  $P(\bm{R}_i \mid \bm{\epsilon}^i, \bm{\omega}^i ) = P(\bm{R}_i \mid \bm{\epsilon}^{\prime i}, \bm{\omega}^{ i} ) $. By this contradiction, we prove $\bm{\beta}$ is identifiable. Using similar arguments, we can show $\bm{\gamma}, \bm{\lambda}$ are also identifiable and hence complete the proof.

Combining \textit{Steps 1--3}, we prove that $(\bm{\beta}, \bm{\gamma}, \bm{\lambda})$ in RegCDMs are identifiable if and only if $\mathbf{J}^0$ has full column rank under Conditions \ref{A1}--\ref{A3}.
\end{proof}

To prove the main results in Section \ref{proposed identifiability conditions}, we next introduce other useful lemmas and corollaries from existing works in literature. Lemma \ref{lemma kruskal 1977} and Corollaries \ref{strict identifiability corollary}--\ref{generic identifiability corollary} summarize the  conditions for the global identifiability of general restricted latent class models proposed by \citeA{allman2009}, which is based on the algebraic results in \citeauthor{kruskal1977} \citeyear{kruskal1977}. 

Before presenting these lemmas and corollaries, we introduce the decomposition of $\bm{\Psi}$ and some notation definitions. The decomposition of $\bm{\Psi}$ is similar as the decomposition of $\bm{\Phi}$ defined in Section \ref{proposed identifiability conditions} in the main text. We divide the total of $J$ items into three mutually exclusive item sets $\mathcal{J}_1, \mathcal{J}_2$ and $\mathcal{J}_3$ containing $J_1, J_2$ and $J_3$ items respectively, with $J_1 + J_2 + J_3 = J$. For $t = 1,2$ and $3$, let $\mathcal{S}_{J_t}$ be the set containing the response patterns from items in $\mathcal{J}_t$ with cardinality of $\mathcal{S}_{J_t}$  to be $\kappa_t = |\mathcal{S}_{J_t}| = \prod_{j \in \mathcal{J}_t} M_j$. The submatrix $\bm{\Psi}_t$ has dimension $\kappa_t \times C$. The definition for the entries in $\bm{\Psi}_t$ is the same as in \eqref{psi definition}, except that each row of $\bm{\Psi}_t$ corresponds to one response patterns $\bm{r} \in \mathcal{S}_{J_t}$ while each row of $\bm{\Psi}$ corresponds to $\bm{r} \in \mathcal{S}^{\prime}$.

 \begin{lemma}\cite{kruskal1977}
 \label{lemma kruskal 1977}
For $t = 1,2$ and $3$, denote $O_t = rank_K( \bm{\Psi}_t)$ as the Kruskal rank of $\bm{\Psi}_t$, where $\bm{\Psi}_t$ is a decomposed matrix of $\bm{\Psi}$.  If
 \[
 O_1 + O_2 + O_3 \geq 2C + 2,
 \]
then $\bm{\Psi}_1, \bm{\Psi}_2$ and $\bm{\Psi}_3$ uniquely determines the decomposition of $\bm{\Psi}$ up to simultaneous permutation and rescaling of columns.
 
 \end{lemma}

 \begin{corollary}\cite{allman2009}
 \label{strict identifiability corollary} 
 Consider the restricted latent class models with $C$ classes. For $t = 1,2$ and $3$, let $\bm{\Psi}_t$ denote a decomposed matrix of $\bm{\Psi}$ and $O_t$ denote its Kruskal rank. If
 \[
 O_1 + O_2 + O_3 \geq 2C + 2,
 \]
then the parameters of the model are uniquely identifiable, up to label swapping.
 \end{corollary}

\begin{corollary}\cite{allman2009}
 \label{generic identifiability corollary}
 Continue with the setting in Corollary \ref{strict identifiability corollary}. For $t = 1,2,3$, let $\bm{\Psi}_t$ denote a decomposed matrix of $\bm{\Psi}$ and $\kappa_t$ denote its row dimension. If
 \[
 \min\{C, \kappa_1\} + \min\{C, \kappa_2\} + \min\{C, \kappa_3\} \geq 2C + 2,
 \]
 Then the parameters of the restricted latent class models are generically identifiable up to label swapping.
\end{corollary}

Combining all these results as well as Proposition 2 in \citeauthor{huang2004} \citeyear{huang2004}, we present Lemma \ref{psi and tau}, which is the key in the proof of Theorem \ref{covariate strict identifiability regLCMs}.

\begin{lemma}
\label{psi and tau}
 For the polytomous-response RegLCMs, $(\bm{\eta}^i, \bm{\Theta}^i)$ are strictly identifiable if Conditions \ref{A1} ,\ref{A2} and \ref{lemma 3.3a} hold, and are generically identifiable if Conditions \ref{A1}, \ref{A2} and \ref{lemma 3.3b} hold.
\begin{enumerate}[label=($B$\arabic*)]
    
	\setcounter{enumi}{2}
	\item The matrix $\bm{\Phi}$ can be decomposed into $\bm{\Phi}_1,$ $\bm{\Phi}_2,$ $\bm{\Phi}_3$, with Kruskal rank of each $\bm{\Phi}_t$ to be $I_t$ and the dimension of each $\bm{\Phi}_t$ to be $\kappa_t \times C$. We have either
    \begin{enumerate}[label=(B3.\alph*)]
    \item \label{lemma 3.3a} $I_1 + I_2 + I_3 \geq 2C + 2$; or
    \item \label{lemma 3.3b}  $\min\{C, \kappa_1\} + \min\{C, \kappa_2\} + \min\{C, \kappa_3\} \geq 2C + 2$.
    \end{enumerate}
\end{enumerate}

\end{lemma}

The proof of Lemma \ref{psi and tau} is provided in Section B.

\begin{proof}[Proof of Theorem \ref{covariate strict identifiability regLCMs}]
\label{proof of theorem1}

From Condition \ref{C4NEW}, the Kruskal rank $I_t$ of $\bm{\Phi}_t$ satisfy the arithmetic condition of Condition \ref{lemma 3.3a} in Lemma \ref{psi and tau}. As we assumed in Theorem \ref{covariate strict identifiability regLCMs}, Conditions \ref{A1} and \ref{A2} also hold. According to Lemma \ref{psi and tau},  RegLCMs are strictly identifiable at $(\bm{\eta}^i$, $\bm{\Theta}^i)$ for $i = 1, \dots, N$. Following the similar arguments in \textit{Steps 2--3} from the \textit{Proof of Theorem \ref{local RegCDMs}}, we show that $(\bm{\beta}, \bm{\gamma}, \bm{\lambda})$ in RegLCMs are identifiable given $(\bm{\eta}^i, \bm{\Theta}^i)$ are identifiable under Condition \ref{A3}. Hence we complete the proof.
\end{proof}

\begin{proof}[Proof of Proposition \ref{covariate strict identifiability CDMs}]
	As mentioned in Section \ref{proposed identifiability conditions}, \ref{C4star} is the sufficient   condition for the identifiability of general restricted latent class models with binary responses according to Theorem 1 in \citeA{xu2017}.  This condition is further extended to restricted latent class models with polytomous responses by Theorem 2 in \citeauthor{Culpepper2019} \citeyear{Culpepper2019}. So for RegCDMs, $(\bm{\eta}^i$, $\bm{\Theta}^i)$  are strictly identifiable given Condition \ref{C4star} for  $i = 1,\dots, N$. Then based on the the similar arguments in \textit{Steps 2--3} from the \textit{Proof of Theorem \ref{local RegCDMs}}, $(\bm{\beta}, \bm{\gamma}, \bm{\lambda})$ in RegCDMs are identifiable given $(\bm{\eta}^i, \bm{\Theta}^i)$ are identifiable.
	\end{proof}
\begin{proof}[Proof of Theorem \ref{generic RegLCMs}]
	For $t = 1,2$ and $3$, the decomposed matrix $\bm{\Phi}_t$ and the decomposed matrix $\bm{\Psi}_t$ have the same row dimension $\kappa_t$. So given Condition \ref{C4prime}, Condition \ref{lemma 3.3b} in Lemma \ref{psi and tau} holds. According to Lemma \ref{psi and tau},  RegLCMs are generically identifiable at $(\bm{\eta}^i$, $\bm{\Theta}^i)$ for $i = 1, \dots, N$.
 Based on the similar arguments in \textit{Steps 2--3} from the \textit{Proof of Theorem \ref{local RegCDMs}}, $(\bm{\beta}, \bm{\gamma}, \bm{\lambda})$ in RegLCMs are generically identifiable given $(\bm{\eta}^i, \bm{\Theta}^i)$ are generically identifiable.
\end{proof}

\begin{proof}[Proof of Proposition \ref{generic CDMs}]
	In Proposition 5.1(b) of \citeauthor{gu2018partial} \citeyear{gu2018partial}, Condition \ref{C4prime2} is  sufficient for the generic identifiability of CDMs. So for RegCDMs, $(\bm{\eta}^i$, $\bm{\Theta}^i)$ are generically identifiable under Condition \ref{C4prime2} for $i = 1, \dots, N$.	Based on the the similar arguments in \textit{Steps 2--3} from the \textit{Proof of Theorem \ref{local RegCDMs}}, $(\bm{\beta}, \bm{\gamma}, \bm{\lambda})$ in RegCDMs are generically identifiable given $(\bm{\eta}^i, \bm{\Theta}^i)$ are generically identifiable.
\end{proof}

\section{Proofs of Lemmas}
\label{Proofs of Lemmas}
\begin{proof}[Proof of Lemma \ref{tau and tau 0}]
For notational convenience, we use $\bm{\eta}$, $\bm{\Theta}$, $\bm{\epsilon}$ and $\bm{\omega}$ to denote the parameters $\bm{\eta}^i$, $\bm{\Theta}^i$, $\bm{\epsilon}^i$ and $\bm{\omega}^i$ of a general subject $i$. According to the definition of identifiability, $(\bm{\eta},\bm{\Theta})$ are identifiable means that there exist no $(\bm{\eta},\bm{\Theta})$ $\neq$ $(\bm{\eta}^{\prime}$, $\bm{\Theta}^{\prime})$ such that $P(\bm{R} = \bm{r} \mid\bm{\eta},\bm{\Theta}) = P(\bm{R} = \bm{r} \mid\bm{\eta}^{\prime}, \bm{\Theta}^{\prime})$. To prove Lemma \ref{tau and tau 0} that $(\bm{\eta}$, $\bm{\Theta})$ are identifiable if and only if $(\bm{\epsilon}$, $\bm{\omega})$ are identifiable, we need to show that the transformation from $(\bm{\eta}$, $\bm{\Theta})$ to $(\bm{\epsilon}$, $\bm{\omega})$ is bijective. We next illustrate this bijective mapping from $\bm{\eta}$ to $\bm{\epsilon}$ holds by showing $({\eta_0}, \cdots, {\eta_{C-1}}) = (\eta_0^{\prime}, \cdots, \eta_{C-1}^{\prime})$ if and only if $({\epsilon_0}, \cdots, {\epsilon_{C-1}}) = (\epsilon_0^{\prime}, \cdots, \epsilon_{C-1}^{\prime})$. 

First, we show that $({\eta_0}, \cdots, {\eta_{C-1}}) = (\eta_0^{\prime}, \cdots, \eta_{C-1}^{\prime})$ implies $({\epsilon_0}, \cdots, {\epsilon_{C-1}}) = (\epsilon_0^{\prime}, \cdots, \epsilon_{C-1}^{\prime})$. For $c = 0, \ldots, C-1$, under the condition that
\[
  \eta_c =  \frac{e^{\epsilon_c}}{1 +\sum_{s = 1}^{C-1} e^{\epsilon_s} } =\frac{e^{\epsilon_c^{\prime}}}{1 +\sum_{s = 1}^{C-1} e^{\epsilon_s^{\prime}} } = \eta_c^{\prime},
\]
we can write
\begin{equation}
e^\delta = \frac{e^{\epsilon_0}}{e^{\epsilon_0^{\prime}}} = \cdots = \frac{e^{\epsilon_c}}{e^{\epsilon_c^{\prime}}} = \cdots =\frac{e^{\epsilon_{C-1}}}{e^{\epsilon_{C-1}^{\prime}}} = \frac{1 +\sum_{s = 1}^{C-1} e^{\epsilon_s} }{1 +\sum_{s = 1}^{C-1} e^{\epsilon_s^{\prime}}}, \nonumber
\end{equation}
where $e^\delta$ denotes the common ratio among all ${e^{\epsilon_c}}/{e^{\epsilon_c^{\prime}}}$. Hence
\begin{equation}
\delta = \epsilon_c - \epsilon_c^{\prime} , \quad  c = 0,\cdots,C-1. \label{all epsilon equal}
\end{equation}
Substituting every $\epsilon_c^{\prime}$ with $\epsilon_c - \delta$ into the equation $\eta_0 = \eta_0^{\prime}$, we have
\begin{equation}
\frac{e^{\epsilon_{0}}}{1+\sum_{s=1}^{C-1} e^{\epsilon_{s}}}=\frac{e^{\epsilon_{0}-\delta}}{1+\sum_{s=1}^{C-1} e^{\epsilon_{s}-\delta}}, \nonumber
\end{equation}
Further simplifying the above equation gives
\begin{equation}
\frac{1}{1+\sum_{s=1}^{C-1} e^{\epsilon_{s}}} = \frac{1}{e^{\delta}+\sum_{s=1}^{C-1} e^{\epsilon_{s}}},  \nonumber
\end{equation}
and then we have
\begin{equation}
e^{\delta}+\sum_{s=1}^{C-1} e^{\epsilon_{s}} \nonumber = 1 + \sum_{s=1}^{C-1} e^{\epsilon_{s}}, \nonumber
\end{equation}
which has unique solution $\delta = 0 $. Taking $\delta = 0 $ back into  \eqref{all epsilon equal}, we have $\epsilon_c = \epsilon_c^{\prime}$ for all $c = 0, \ldots, C-1$. Therefore $\bm{\epsilon} = (\epsilon_0, \cdots, \epsilon_{C-1}) $ is equivalent to $\bm{\epsilon^{\prime}} = (\epsilon_0^{\prime}, \cdots, \epsilon_{C-1}^{\prime})$.  

Next we prove $({\epsilon_0}, \cdots, {\epsilon_{C-1}}) = (\epsilon_0^{\prime}, \cdots, \epsilon_{C-1}^{\prime})$ implies $({\eta_0}, \cdots, {\eta_{C-1}}) = (\eta_0^{\prime}, \cdots, \eta_{C-1}^{\prime})$. This part is straightforward as $({\epsilon_0}, \cdots, {\epsilon_{C-1}}) = (\epsilon_0^{\prime}, \cdots, \epsilon_{C-1}^{\prime})$ implies that for any $c = 0, \ldots, C-1$, we have
\begin{equation}
	 \frac{\exp(\epsilon_{c})}{1 + \sum_{s = 1}^{C-1} \exp(\epsilon_{s})} =  \frac{\exp(\epsilon_{c}^{\prime})}{1 + \sum_{s = 1}^{C-1} \exp(\epsilon_{s}^{\prime})}. \nonumber
\end{equation}
Equivalently, we show $\eta_c = \eta_c^{\prime}$ for any $c = 0, \ldots, C-1$. So $({\eta_0}, \cdots, {\eta_{C-1}}) = (\eta_0^{\prime}, \cdots, \eta_{C-1}^{\prime})$. Combining the above arguments, we prove $\bm{\eta} = \bm{\eta^{\prime}} \ $if and only if $\bm{\epsilon} = \bm{\epsilon^{\prime}}$.

Similar arguments can be applied to show $\bm{\Theta} = \bm{\Theta^{\prime}}$ if and only if $\bm{\omega} = \bm{\omega^{\prime}}$.  Hence $(\bm{\eta}$, $\bm{\Theta})$ are identifiable if and only if $(\bm{\epsilon}$, $\bm{\omega})$ are identifiable.
\end{proof}

	\begin{proof}[Proof of Lemma \ref{lemma_in_thm2_step1}]
	We  prove the the first part, that is, $\mathbf{J}^0$ has full column rank if and only if $\overline{\mathbf{J}}^{0}$
		 are linearly independent. The second part regarding $\mathbf{J}^i$ can be similarly proved.
		  
		  To show the linear independence of $\mathbf{J}^0$ or $\overline{\mathbf{J}}^{0}$, we need to establish the relationship between the two linear combinations as follows.
		  For any linear combinations of the columns in $\mathbf{J}^0$ with coefficients $h_c^0$'s and $l_{jrc}^0$'s, there exist $a_c^0$'s and $b_{jrc}^0$'s such that the following equation holds.
\begin{eqnarray}
&& \sum_{c=1}^{C-1}h_{c}^0(\bm{\psi}_{c}^0- \bm{\psi}_0^0) + \sum_{j=1}^{J} \sum_{r=1}^{M_j-1} \sum_{c=0}^{C-1}\left\{l_{jrc}^0 \eta_{c}^0(\frac{\mathbf{I}\{{r}_{j} = r\} }{\theta_{jrc}^0} - \frac{\mathbf{I}\{{r}_{j} = 0\} }{\theta_{j0c}^0}   ) \odot \bm{\psi}_{c}^0 \right\} \label{li J0} \\
&& = 	\sum_{c=0}^{C-1}a_{c}^0\bm{\psi}_{c}^0+ \sum_{j=1}^{J} \sum_{r=0}^{M_j-1} \sum_{c=0}^{C-1}\left\{b_{jrc}^0 \eta_{c}^0(\frac{\mathbf{I}\{{r}_{j} = r\} }{\theta_{jrc}^0}  ) \odot \bm{\psi}_{c}^0 \right\}, 	\label{trans li J0}
\end{eqnarray}
where \begin{equation}
\label{a to h}
	a_c^0 =
  \begin{cases}
    h_c^0, & \text{if $c \neq 0$},  \\
     - (h_1^0 + \dots + h_{C-1}^0), & \text{if $c = 0$}, \\
  \end{cases}   \\
\end{equation}
and for any $j = 1,\dots, J$, $c= 0, \dots, C-1$,
\begin{equation}
\label{b to l}
  b_{jrc}^0 =
  \begin{cases}
    l_{jrc}^0, & \text{if $r \neq 0, $} \\
     - (l_{j1c}^0 + \dots + l_{j(M_j-1)c}^0), & \text{if $r = 0$}.  \\
  \end{cases}  
\end{equation}
With the above relationship established, we next show that $\mathbf{J}^0$ has full column rank if and only if $\overline{\mathbf{J}}^{0}$ are linearly independent. When   $\overline{\mathbf{J}}^{0} $ are linearly independent, \eqref{trans li J0} = $\bm{0}$ implies $a_c^0 = 0$ and $b_{jrc}^0 = 0$, which further implies $h_c^0 = 0$ and $l_{jrc}^0 = 0$ by \eqref{a to h} and \eqref{b to l}. So \eqref{li J0} = $\bm{0}$ implies $h_c^0 = 0$ and $l_{jrc}^0 = 0$. Hence, $\mathbf{J}^0$ has full column ranks. Similarly, when $\mathbf{J}^0$ has full column ranks, \eqref{li J0} = $\bm{0}$ implies $h_c^0 = 0$ and $l_{jrc}^0 = 0$ which further implies $a_c^0 = 0$ and $b_{jrc}^0 = 0$ by \eqref{a to h} and \eqref{b to l}. So \eqref{trans li J0} = $\bm{0}$ implies $a_c^0 = 0$ and $b_{jrc}^0 = 0$. Hence, $\overline{\mathbf{J}}^{0}$ are linearly independent.	\end{proof}

\begin{proof}[Proof of Lemma \ref{psi and tau}]
This proof is motivated from the \textit{Proof of Proposition 2} in \citeauthor{huang2004} \citeyear{huang2004}. Before presenting the proof, we set up a few notations. In Section \ref{proposed identifiability conditions}, $\bm{\Phi}$ can be decomposed into $\bm{\Phi}_1, \bm{\Phi}_2$ and $\bm{\Phi}_3$, where each $\bm{\Phi}_t$ has Kruskal rank $I_t$ and row dimension $\kappa_t$. 
And in Appendix \ref{Proofs of Propositions and Theorems}, $\bm{\Psi}^i$ can be decomposed into $\bm{\Psi}_1^i, \bm{\Psi}_2^i$ and $\bm{\Psi}_3^i$, where each $\bm{\Psi}_t^i$ has Kruskal rank $O_t^i$ and the same row dimension $\kappa_t$ as $\bm{\Phi}_t$. Denote the columns in $\bm{\Phi}_t$ to be $\bm{\phi}_{t0}, \cdots, \bm{\phi}_{t(C-1)}$ and the columns in $\bm{\Psi}_t^i$ to be $\bm{\psi}_{t0}^i, \cdots, \bm{\psi}_{t(C-1)}^i$. Conditions \ref{A1} and \ref{A2} are shown to be necessary in Section \ref{Models and Existing Works} and assumed to hold. To prove Condition \ref{lemma 3.3a} is sufficient for the strict identifiability of $(\bm{\eta}^i$, $\bm{\Theta}^i)$, we first need to show that for $t = 1,2$ and $3$, given $\bm{\Phi}_t$ has Kruskal rank $I_t$, the equation $O_t^i \geq I_t$ holds, so that $I_1 + I_2 + I_3 \geq 2C + 2$ from Condition \ref{lemma 3.3a} implies  $O_1^i + O_2^i + O_3^i \geq 2C + 2$. Then based on Corollary  \ref{strict identifiability corollary} that $(\bm{\eta}^i$, $\bm{\Theta}^i)$ are strictly identifiable under the condition that $O_1^i + O_2^i + O_3^i \geq 2C + 2$, we complete the proof of strict identifiability.

The remaining part is to show $O_t^i \geq I_t$ for $t = 1,2$ and $3$ and for $i = 1, \dots, N$. Without loss of generality, we only show $O_1^i \geq I_1$, then $O_2^i \geq I_2$ and $O_3^i \geq I_3$ can be similarly proved.  Under the condition that any set of $I_1$ columns in $\bm{\Phi}_1$ are linearly independent, $\bm{\phi}_{1\sigma(1)}, \cdots ,\bm{\phi}_{1\sigma({I_1})}$ are linearly independent for any permutation $\sigma$ on $\{1,\dots, I_1\}$ such that $ \{\sigma(1), \ \sigma(2), \cdots, \ \sigma(I_1)\} \subseteq \{0,\cdots, C-1\}$. To show $O_t^i \geq I_t$, we need $\bm{\psi}_{1\sigma(1)}^i, \cdots ,\bm{\psi}_{1\sigma({I_1})}^i$ to be linearly independent for any permutation set $ \{\sigma(1), \ \sigma(2), \cdots, \ \sigma(I_1)\}$. The linear combinations of $\bm{\phi}_{1\sigma(1)}, \cdots ,\bm{\phi}_{1\sigma({I_1})}$ can be associated with the linear combinations of $\bm{\psi}_{1\sigma(1)}^i, \cdots ,\bm{\psi}_{1\sigma({I_1})}^i$ as follows. For any permutation $\sigma$ and $a_{\sigma(c)}$, there exists $b_{\sigma(c)}$ and $\bm{Y_1}^i$ such that 
\begin{equation}
\label{ac acstar}
	\sum_{c=1}^{I_1} a_{\sigma(c)} \bm{\psi}_{1\sigma(c)}^i = (\sum_{c=1}^{I_1} b_{\sigma(c)}\bm{\phi}_{1\sigma(c)})\odot \bm{Y_1}^i      
\end{equation}
where \begin{eqnarray}
	 \bm{Y_1}^i &=& \left( \prod_{j \in \mathcal{J}_1} \text{exp}{(\lambda_{1jr_j}z_{ij1} + \cdots + \lambda_{qjr_j}z_{ijq}}) : \bm{r} = (r_1,\dots, r_J) \in \mathcal{S}_{J_1}\right)_{\kappa_1 \times 1}, \nonumber \\
b_{\sigma(c)} &=& a_{\sigma(c)} \prod_{j \in \mathcal{J}_1} \frac{1 + \sum_{s = 1}^{M_j-1 } e^{\gamma_{js\sigma(c)}}}{ 1 + \sum_{s = 1}^{M_j - 1}\exp (\gamma_{js\sigma(c)} + \lambda_{1js}z_{ij1} + \cdots + \lambda_{qjs}z_{ijq})}. \label{bsigma asigma}
\end{eqnarray}
To show $\bm{\psi}_{1\sigma(1)}^i, \cdots ,\bm{\psi}_{1\sigma({I_1})}^i$ to be linearly independent, we need to show $\sum_{c=1}^{I_1} a_{\sigma(c)} \bm{\psi}_{1\sigma(c)}^i = \bm{0}$ implies $a_{\sigma(c)} =0$ for any $\sigma$. Based on \eqref{ac acstar}, we have $\sum_{c=1}^{I_1} a_{\sigma(c)} \bm{\psi}_{1\sigma(c)}^i = \bm{0}$ implies $\sum_{c=1}^{I_1} b_{\sigma(c)}\bm{\phi}_{1\sigma(c)} = \bm{0}$.  Under the condition that $\bm{\phi}_{1\sigma(1)}, \cdots ,\bm{\phi}_{1\sigma({I_1})}$ are linear independent, $\sum_{c=1}^{I_1} b_{\sigma(c)}\bm{\phi}_{1\sigma(c)} = \bm{0}$ implies $b_{\sigma(1)} = \cdots =  b_{\sigma(I_1)} = 0$. And by \eqref{bsigma asigma}, $a_{\sigma(1)} = \cdots = a_{\sigma(I_1)} =0$. Hence $\bm{\psi}_{1\sigma(1)}^i, \cdots ,\bm{\psi}_{1\sigma({I_1})}^i$ are linearly independent for any $\sigma$. Hence we show $O_1^i \geq I_1$ for $i = 1, \ldots, N$ and complete the proof for strict identifiability.

For Condition \ref{lemma 3.3b}, because each $\bm{\Psi}_t^i$ has row dimension $\kappa_t$ the same as $\Phi_t$ does and we have $\min\{C, \kappa_1\} + \min\{C, \kappa_2\} + \min\{C, \kappa_3\} \geq 2C + 2$, according to Corollary \ref{generic identifiability corollary},  $(\bm{\eta}^i$, $\bm{\Theta}^i)$ are generically identifiable under Condition \ref{lemma 3.3b} for $i = 1, \ldots, N$. 
\end{proof}

\end{document}